\documentclass[11pt]{article}
\usepackage[margin=1in]{geometry}
\usepackage{amsmath, amsfonts, amsthm}
\usepackage{natbib}
\usepackage{pdflscape}
\newcommand{\E}{\mathbb{E}}
\newcommand{\R}{\mathbb{R}}
\renewcommand{\P}{\mathbb{P}}
\DeclareMathOperator*{\argmin}{argmin} 
\DeclareMathOperator*{\argmax}{argmax} 

\DeclareMathOperator*{\Lip}{Lip} 
\DeclareMathOperator*{\bdd}{bdd} 
\DeclareMathOperator*{\nn}{nn} 
\DeclareMathOperator*{\id}{1}
\DeclareMathOperator*{\var}{var}
\DeclareMathOperator*{\unid}{2}
\usepackage{optidef}
\usepackage{booktabs}
\usepackage{enumerate}

\usepackage{pgfplots}

\newtheorem{theorem}{Theorem}
\newtheorem{lemma}{Lemma}
\theoremstyle{definition}
\newtheorem{assumption}{Assumption}

\title{Off-policy evaluation beyond overlap: partial identification through smoothness}

\author{%
	Samir Khan \\
  Stanford University\\
  \texttt{samirk@stanford.edu} \\
  \and
   Martin Saveski \\
   Stanford University \\
   \texttt{msaveski@stanford.edu} \\
  \and
   Johan Ugander \\
   Stanford University \\
   \texttt{jugander@stanford.edu} \\
}
\begin{document}
\maketitle

\begin{abstract}

	Off-policy evaluation (OPE) is the problem of estimating the value of a target policy using historical data collected under a different logging policy. OPE methods typically assume overlap between the target and logging policy, enabling solutions based on importance weighting and/or imputation. In this work, we approach OPE without assuming either overlap or a well-specified model by considering a strategy based on partial identification under non-parametric assumptions on the conditional mean function, focusing especially on Lipschitz smoothness. Under such smoothness assumptions, we formulate a pair of linear programs whose optimal values upper and lower bound the contributions of the no-overlap region to the off-policy value. We show that these linear programs have a concise closed form solution that can be computed efficiently and that their solutions converge, under the Lipschitz assumption, to the sharp partial identification bounds on the off-policy value. Furthermore, we show that the rate of convergence is minimax optimal, up to log factors. We deploy our methods on two semi-synthetic examples, and obtain informative and valid bounds that are tighter than those possible without smoothness assumptions.

\end{abstract}


\section{Introduction}

Off-policy evaluation (OPE) is the problem of estimating the value of an evaluation/target policy using data from a behavior/logging policy, and arises naturally in many settings including medicine, program evaluation, and recommender system design \citep{li2011, bottou2013, swaminathan2017, liao2021, chin2022}. Generally speaking, methods for off-policy evaluation typically take one of two forms: reweighting or imputation. Reweighting methods, as the name suggests, compute the ratio of action probabilities under the evaluation policy and the behavior policy, and then use these weights to obtain unbiased estimates of the value of the evaluation policy. Imputation methods, on the other hand, model the outcome observed when taking an action as a function of covariates, use this model to estimate outcomes under the evaluation policy, and then average the estimated outcomes to obtain an off-policy value estimate. Finally, doubly-robust methods combine these two approaches to obtain better theoretical guarantees \citep{dudik2011}.

However, all of these approaches require overlap between the behavior policy and the evaluation policy. That is, there must not be actions where the behavior policy assigns a positive probability but the evaluation policy assigns zero probability. If there are such actions, then the weights used by reweighting methods will be infinite, and the models used by imputation methods will only be valid under the strong and unverifiable assumption that they can accurately extrapolate from the support of the behavior policy (where they were trained) to the support of the evaluation policy (where they must predict). This extrapolation can be justified for well-specified parametric models, but will lead to biased estimates in general \citep{sachdeva2020}.

The overlap requirement fundamentally limits the utility of off-policy evaluation, since if we wish to explore a large space of potential evaluation policies, we must have a behavior policy that takes every action with positive probability. Universal overlap can be achieved by having, for example, a uniform behavior policy, but in high-stakes settings where some actions are undesirable, running a uniform behavior policy can be ethically or financially prohibitive. In these settings, we would like to deploy a behavior policy that assigns zero probability to some actions and still be able to evaluate policies that may take those actions. 

In this work, we address these challenges by developing methods for off-policy evaluation that handle overlap violations while making only weak, non-parametric assumptions. Our approach is based on the observation that overlap violations are essentially an identifiability problem, and so can be handled using partial identification methods that report bounds on the off-policy value under assumptions on the underlying data generating distribution \citep{manski1990, manski2010}. We focus mainly on smoothness assumptions that constrain the conditional outcome to be $L$-Lipschitz, and also briefly discuss monotonicity assumptions on the conditional outcome. Approaches based on stronger smoothness assumptions, such as smoothness of higher derivatives of the conditional outcome or of H\"older continuity, are interesting directions for future work. All of these smoothness assumptions are generalizations of the classical bounded response assumption used to obtain partial identification in \citet{manski1990}. Under the Lipschitz assumptions, we provide bounds on the off-policy value rather than point estimates, thus accurately reflecting our uncertainty about the model's ability to extrapolate into the no-overlap region. 

\begin{figure}[t]
	\centering
	\includegraphics{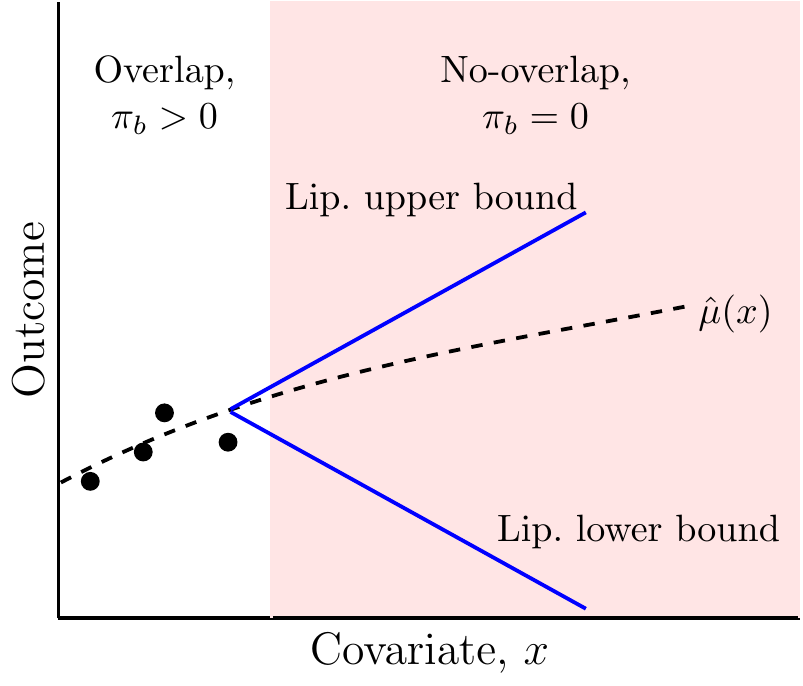}
	\caption{A visualization of our approach on a toy problem with three data points in a one-dimensional covariate space. The data are collected using a behavior policy $\pi_b$. The function $\hat{\mu}(x)$ (dashed black) is estimated based on the points observed in the overlap region where $\pi_b>0$ and treatment probability is positive, so we are unsure whether it is accurate in the no-overlap region where $\pi_b=0$ and treatment probability is zero (red). Rather than directly using the predictions of $\hat{\mu}$ in the no-overlap region, we assume the true outcome is $L$-Lipschitz and use this assumption to provide upper and lower bounds.}
\label{fig:demo}
\end{figure}

The spirit of this approach is illustrated in Figure~\ref{fig:demo}, which shows a toy problem with a single real-valued covariate. For small values of the covariate, we have overlap and so the behavior policy treats units with small covariate values with positive probabilities. For those treated units, we observe their outcomes (shown as black dots). On the other hand, for large values of the covariate, we have no-overlap (indicated by red), and so units in this region are never treated, and we do not observe any responses. We can fit a model $\hat{\mu}(x)$ (shown as a dashed black line) to the observed responses, but because this model has not been trained on observations in the no-overlap region, there is no way to guarantee that its predictions there are accurate. Thus, rather than directly using the model predictions, we assume that the true conditional outcome is $L$-Lipschitz. Under this assumption, if the model $\hat{\mu}$ is consistent in the overlap region, we can give bounds (shown in blue) on the outcome in the no-overlap region. These bounds on the outcome lead to bounds on the off-policy value and constitute our partial identification region. We emphasize that smoothness assumption here is with respect to a particular covariate space and metric on that covariate space; as such, the covariate space and metric should be chosen based on domain knowledge about which covariates and metric the outcome is plausibly smooth with respect to.

To reiterate, the key idea is that we are replacing an assumption on the correctness of $\hat{\mu}$ with an assumption on the smoothness of the true outcome function, and thus replacing estimates of the conditional outcome function with bounds on the conditional outcome function. An appealing feature of this approach is that the choice of Lipschitz parameter $L$ directly modulates the strength of the smoothness assumptions. Thus, varying $L$ leads to a sensitivity analysis that can highlight trade-offs between the strength of our assumptions and the strength of the conclusions we can draw.

\subsection{Related work}

The overlap assumption is standard in off-policy evaluation literature \citep{bottou2013, swaminathan2015, thomas2016, wang2017} and we focus our discussion on the emerging literature on policy evaluation and learning in the no-overlap setting. 

One strand of this literature relies on access to a well-specified outcome model. For example, \citet{mou2023} assume that the reward function lies within a reproducing kernel Hilbert space and use this assumption to extrapolate into the no-overlap region, thus effectively assuming a well-specified outcome model. Additionally, their method requires the action probabilities to be positive, albeit potentially arbitrarily small, while our results allow for action probabilities that are exactly zero. On the policy learning side, \citet{sachdeva2020} first show that a policy learned using inverse-propensity scoring is sub-optimal, and then consider several remedies based on restricting the set of policies that are optimized over or using a well-specified model. In contrast to these works, we require no such extrapolation and allow for evaluation of arbitrary policies.

Another strand of the no-overlap literature develops methods for policy learning by optimizing lower bounds on the off-policy value over some class of outcome functions. This is the approach taken by \citet{higbee2022}, \citet{manskiusing}, and \citet{ben2021}, although in all cases the bounds they optimize are potentially quite loose and are thus inappropriate for providing information about the value of a particular policy; we demonstrate this looseness in simulations in Section~\ref{sec:experiments}. Of these works, \citet{higbee2022} and \citet{manskiusing} are slightly further from ours, in that they make assumptions on the expected outcome as a function of action rather than as a function of covariates, and \citet{ben2021} is slightly closer to ours, in that they make a similar Lipschitz assumption on the expected outcome as a function of covariates. In all cases, our estimators are different and the theoretical optimality results we obtain for our estimators go beyond the results found in these works and more completely develop the partial identification framework for off-policy evaluation under smoothness; see Appendix~\ref{sec:ben_compare} for further details.

In the reinforcement learning literature, \citet{jiang2020} propose a ``minimax value interval'' that is valid even under no-overlap. However, their method relies on a well-specified outcome model, requires solving a challenging min-max optimization problem, and is more relevant when both the outcome model and behavior policies are unknown. In our setting, the behavior policy is known, and so our methods, which also come with optimality guarantees, are preferable.

Finally, \citet{saveski2023} deploy methods based on partial identification under smoothness assumptions in an off-policy evaluation problem arising in the context of matching reviewers to papers at academic conferences. However, they combine the steps of fitting a model $\hat{\mu}$ and estimating bounds on the off-policy value into a single linear program, leading to sub-optimal bounds. Their work highlights the importance of the no-overlap setting, and thus of the potential for our methods and results to provide stronger conclusions on existing data sets.


\section{Model and notation}
\label{sec:model}

We consider the problem of off-policy evaluation of stochastic policies that select an action from an action space based on covariates $X_i$. The covariates $X_i$ take values in a metric space $\mathcal{X}$ with metric $d(\cdot, \cdot)$. The actions are random variables $A_i$ that takes values in a set $\mathcal{A}$. Each action $a\in \mathcal{A}$ has a corresponding potential outcome $Y_i(a)$, and we assume that tuples $(X_i, \{Y_i(a)\}_{a\in \mathcal{A}})$ are i.i.d.~from a distribution $P_0$ on $\mathcal{X}\times \R^{|\mathcal{A}|}$ for some covariate space $\mathcal{X}$. The actions $A_i$ are generated such that $\P(A_i=a\mid X_i=x)=\pi_b(x, a)$ for a behavior policy $\pi_b:\mathcal{X}\times \mathcal{A}\to [0,1]$ that satisfies the constraint $\sum_{a\in \mathcal{A}} \pi_b(x, a)=1$, for all $x\in \mathcal{X}$. These constraints ensure that $\pi_b$ maps covariates $X_i$ to a probability distribution over $\mathcal{A}$. 

In this set-up, we observe $(X_1,A_1, Y_1(A_1)),\ldots , (X_n, A_n, Y_n(A_n))$ under the  behavior policy $\pi_b$ and would like to estimate the value of a different policy $\pi_e$. That is, we would like to estimate the functional 
\begin{equation}
	\psi(P_0)=\E_{(X_i, Y_i)\sim P_0, \hspace{0.8mm} A_i\mid X_i\sim \pi_e}[Y_i(A_i)],
\end{equation}
where $A_i|X_i$ is drawn according to $\pi_e$. It is convenient to decompose this functional across the actions and write it as 
\begin{equation}
	\psi(P_0)=\sum_{a\in \mathcal{A}} \E_{(X_i, Y_i)\sim P_0}[Y_i(a)\pi_e(X_i, a)].
	\label{eq:ope_val}
\end{equation}
Critically, in this work, we allow for the possibility that there exists $x\in \mathcal{X}$ and $a\in \mathcal{A}$ such that $\pi_b(x,a)=0$. We refer to $\{x: \pi_b(x,a)=0\}$ as the {\it no-overlap region} and to $\{x:\pi_b(x,a)>0\}$ as the {\it overlap region}. 

The model described thus far corresponds to a general multi-action off-policy evaluation problem. However, by virtue of the decomposition in \eqref{eq:ope_val}, we can naturally reduce any multi-action OPE problem to a binary action OPE problem. For any action $a\in\mathcal{A}$, we can define the binary action $\tilde{A}_i=\mathbf{1}\{A_i=a\}$ and binary evaluation policy $\tilde{\pi}_e(X_i)=\pi_e(X_i, a)$. Then, if we can estimate the functional 
\begin{equation}
	\tilde{\psi}(P_0)=\E_{P_0}[Y_i(a)\tilde{\pi}_e(X_i)],
	\label{eq:binary_ope_val}
\end{equation}
from data $(X_1, \tilde{A}_1, Y_1(a)\tilde{A}_1),\cdots, (X_n, \tilde{A}_n, Y_n(a)\tilde{A}_n)$, we can combine these estimates across all $a\in\mathcal{A}$ to estimate $\psi(P_0)$.
Further, in a binary problem, we write $\mu_P(x)=\E_P[Y_i(a)\mid X_i=x]$ for the conditional mean function under a distribution $P$. 



\section{Nonparametric partial identification bounds}
\label{sec:estimator}

In this section, we describe our framework for OPE without overlap and then provide several specific instantiations. We work in the setting where $A_i$ is binary; by the reduction of Section~\ref{sec:model}, our results extend naturally to the multi-action setting. For an OPE problem with binary actions we show how to partially identify the off-policy value $\psi(P_0)$ using the assumption that $P_0\in\mathcal{P}$ for a family of distributions $\mathcal{P}$. This family $\mathcal{P}$ encodes the nature of our assumptions on $P_0$ and may take several forms, but one crucial feature that $\mathcal{P}$ has to satisfy is that it must only contain distributions that are consistent with the true distribution $P_0$. 


To elaborate, for any $P\in \mathcal{P}$, the distribution of the observed data $(X_i, A_i, A_iY_i)$ must be the same as it is under the true data generating distribution $P_0$. If this is not the case, then the bounds we obtain under the assumption that $P_0\in \mathcal{P}$ may correspond to a distribution that we know could not have generated the data. More precisely, we say that $P$ is consistent with $P_0$ if (i)~the joint distribution of $(X_i, Y_i)$ under $P$ in the overlap region is the same as the joint distribution of $(X_i, Y_i)$ under $P_0$, and (ii)~the marginal distribution of $X_i$ in the no-overlap region under $P_0$ is the same as the marginal distribution of $X_i$ in the no-overlap region under $P$. 

We describe our approach for a generic family $\mathcal{P}$ and then turn to several concrete choices of $\mathcal{P}$. The first step is to decompose the functional $\psi$ as 
\begin{equation}
	\psi(P)=\underbrace{\E_P[Y_i\pi_e(X_i)\mathbf{1}\{\pi_b(X_i)>0\}]}_{\psi_{\id}} +\underbrace{\E_P[Y_i\pi_e(X_i)\mathbf{1}\{\pi_b(X_i)=0\}]}_{\psi_{\unid}}.
	\label{eq:no_overlap_split}
\end{equation}
so that $\psi_{\id}$ is the contribution of the overlap region and $\psi_{\unid}$ is the contribution of the no-overlap region. This separates out the part of $\psi$ that is identifiable, $\psi_{\id}$, from the part that is not, $\psi_{\unid}$.

The first term, $\psi_{\id}$, is identifiable, so we can estimate it, e.g., using an inverse-probability weighted (IPW) estimator
\begin{equation}
	\hat{\psi}_{\id}(P_0)=\frac{1}{n}\sum_{i=1}^n Y_iA_i\frac{\pi_e(X_i)}{\pi_b(X_i)}\mathbf{1}\{\pi_b(X_i)>0\}.
	\label{eq:psi_1_hat}
\end{equation}
In fact, a self-normalized or doubly-robust estimator can also be used to estimate $\psi_{\id}$ without any modifications to our results \citep{dudik2011, swaminathan2015}; we present the IPW case here for simplicity.

The second term, $\psi_{\unid}$, however is not identified. Our approach is to bound its contribution to \eqref{eq:no_overlap_split} using the assumption that $P_0\in \mathcal{P}$ for some family $\mathcal{P}$. Under this assumption, the tightest possible bounds we could obtain are
\begin{equation}
	\psi_2^-:= \inf_{P\in \mathcal{P}} \psi_{\unid}(P)\quad\text{and}\quad \psi_2^+:=\sup_{P\in\mathcal{P}}\psi_{\unid}(P),
	\label{eq:inf_sup}
\end{equation}
and so we must construct estimators $\hat{\psi}_{\unid}^-$ and $\hat{\psi}_{\unid}^+$ of $\psi_2^-$ and $\psi_2^+$ respectively.


Once we have such estimators, we can set
\begin{equation}
	\hat{\psi}^-=\hat{\psi}_{\id}+\hat{\psi}_{\unid}^-,\quad \hat{\psi}^+=\hat{\psi}_{\id}+\hat{\psi}_{\unid}^+,
	\label{eq:pid_interval}
\end{equation}
and use $[\hat{\psi}^-, \hat{\psi}^+]$ as an interval estimate of $\psi(P_0)$. The following result, whose proof appears in Appendix~\ref{sec:proofs}, guarantees the validity of this interval under conditions on $\hat{\psi}_{\id}, \hat{\psi}_{\unid}^+,$ and $\hat{\psi}_{\unid}^-$.
\begin{theorem}
	\label{thm:interval_consistency}
	Suppose that $\hat{\psi}_{\id}$ is a consistent estimator of $\psi_{\id}(P_0)$, and that $\hat{\psi}_{\unid}^-$ and $\hat{\psi}_{\unid}^+$ are consistent estimators of $\psi_2^-$ and $\psi_2^+$, respectively. Then, for any $\epsilon>0$,
	\begin{equation}
		\lim_{n\to \infty} \P(\hat{\psi}^--\epsilon\leq \psi(P_0)\leq \hat{\psi}^++\epsilon )=1.
		\label{eq:interval_consistency}
	\end{equation}
\end{theorem}
The main challenge is to construct $\hat{\psi}_{\unid}^-, \hat{\psi}_{\unid}^+$ that satisfy the conditions of Theorem~\ref{thm:interval_consistency}. If those conditions are met, then the interval $[\hat{\psi}^-, \hat{\psi}^+]$ will be consistent for $\psi(P_0)$ in the sense of \eqref{eq:interval_consistency}, and so $[\hat{\psi}^-, \hat{\psi}^+]$ will provide valid bounds on $\psi(P_0)$ in large samples. Next, we consider several specific choices of $\mathcal{P}$ that arise from different assumptions on $P_0$, and show how these bounds can be computed. We focus throughout on the infimum, but all of our discussion holds \emph{mutatis mutandis} for the supremum. 
\subsection{Boundedness assumptions}
As a first example, we consider the following simple choice of $\mathcal{P}$, corresponding to the assumption that the response $Y_i$ must lie in the interval $[\ell, u]$:
\begin{equation}
	\mathcal{P}_{\ell, u}^{\bdd}=\left\{ P\text{ consistent w.~}P_0: \ell\leq Y_i\leq u \text{ a.s.}\right\}.
	\label{eq:bdd_fam}
\end{equation}
The assumption that $\ell \leq Y_i\leq u$ implies that $\ell \leq \mu_P(x) \leq u$ for all $x$ as well, and so a natural choice of $\psi_2^-$ is 
\begin{equation}
\hat{\psi}_{\unid}^-=\frac{\ell}{n}\sum_{i=1}^n \pi_e(X_i)\mathbf{1}\{\pi_b(X_i)=0\}
	\label{eq:manski_inf}
\end{equation}
By the law of large numbers, $\hat{\psi}_{\unid}^-\xrightarrow{\P} \E[\ell \pi_e(X_i)\mathbf{1}\{\pi_b(X_i)=0\}]$, and we can verify that this is also the value of $\psi_2^-$. Thus the condition of Theorem~\ref{thm:interval_consistency} holds, and the interval $[\hat{\psi}^-, \hat{\psi}^+]$ is consistent for $\psi(P_0)$. The interval $[\hat{\psi}^-, \hat{\psi}^+]$ is an analogue of the so-called Manski bounds \cite{manski1990}, and so our framework generalizes this well-established practice. As such, we can also obtain confidence intervals for the partial identification region of $\psi(P_0)$ using the methods of \citet{imbens2004} and \citet{stoye2009}.




\subsection{Smoothness assumptions}
\label{subsec:smooth}

Next, we move on to our main focus: Lipschitz assumptions on $\mu_P(x)$. Formally, this corresponds to the family
\begin{equation}
	\mathcal{P}_{L}^{\text{Lip}}=\left\{ P\text{ consistent w.~} P_0: |\mu_P(x_1)-\mu_P(x_2)|\leq Ld(x_1,x_2)\text{ for all }x_1, x_2\in\mathcal{X} \right\},
	\label{eq:smooth_fam}
\end{equation}
where $d$ is a metric on $\mathcal{X}$.
By restricting $\mu_P$ to be $L$-Lipschitz, we can draw conclusions about the behavior of $\mu_P$ in the no-overlap region based on our observations in the overlap region and thus estimate $\inf_{P\in \mathcal{P}_L^{\Lip}} {\psi}_{\unid}(P)$. 

Note that the assumption made in \eqref{eq:smooth_fam} is on the regression function $\mu_P$, rather than on the observed $Y_i$. If we were to make a smoothness assumption on the observed $Y_i$, it would be very difficult to satisfy such an assumption for large datasets, since we would eventually observe pairs of points $(X_i, Y_i)$ and $(X_j, Y_j)$ for which $X_i$ and $X_j$ are close, but $Y_i$ and $Y_j$ are far apart due to observation noise. We would thus want to introduce a level of slack in the Lipschitz assumption corresponding to the noise level of the problem, which amounts to assuming smoothness on a denoised version of the observed $Y_i$. This denoised version of the $Y_i$ is exactly $\mu_P$, and so we would recover the assumption that $P\in \mathcal{P}_L^{\Lip}$.

We propose to construct the estimator $\hat{\psi}_2^-$ in this setting by solving the following linear program:
\begin{mini}|s|
	{t_1,\cdots, t_n}{\frac{1}{n}\sum_{i=1}^n t_i\pi_e(X_i)\mathbf{1}\{\pi_b(X_i)=0\}}
					{}{}
					\addConstraint{|t_i-t_j|\leq Ld(X_i, X_j),\quad 1\leq i<j\leq n}
					\addConstraint{t_i-\hat{\mu}(X_i)=0,\quad 1\leq i\leq n\text{ s.t. }\pi_b(X_i)>0}.
					\label{eq:estim_lp}
				\end{mini}
				The problem \eqref{eq:estim_lp} is an approximation of the population problem $\inf_{P\in \mathcal{P}_{L}^{\Lip}} \psi_{\unid}(P)$ in three ways: \textit{(i)}~it averages over sample points in the objective rather than over $P_0$; \textit{(ii)} it only enforces the Lipschitz constraint between pairs of observed data points $X_i, X_j$ rather than between all pairs $x_1, x_2$ as in \eqref{eq:smooth_fam}; and \textit{(iii)} it sets points in the overlap region to have value $\hat{\mu}$ rather than $\mu$. We will see shortly that all of these approximations are asymptotically negligible.

				We now characterize the solution to \eqref{eq:estim_lp} and its properties under assumptions. The key to our results is the surprising fact that \eqref{eq:estim_lp} can be solved in closed-form whenever $d$ is a metric, even though this is not generally the case for linear programs. The reason we obtain a closed-form solution to \eqref{eq:estim_lp} is because its constraints satisfy what we refer to as a \emph{no-interaction property}, by which we mean that points in the no-overlap region do not place sharp bounds on each other. 

\begin{figure}[t]
	\centering
	\includegraphics[width=0.55\textwidth]{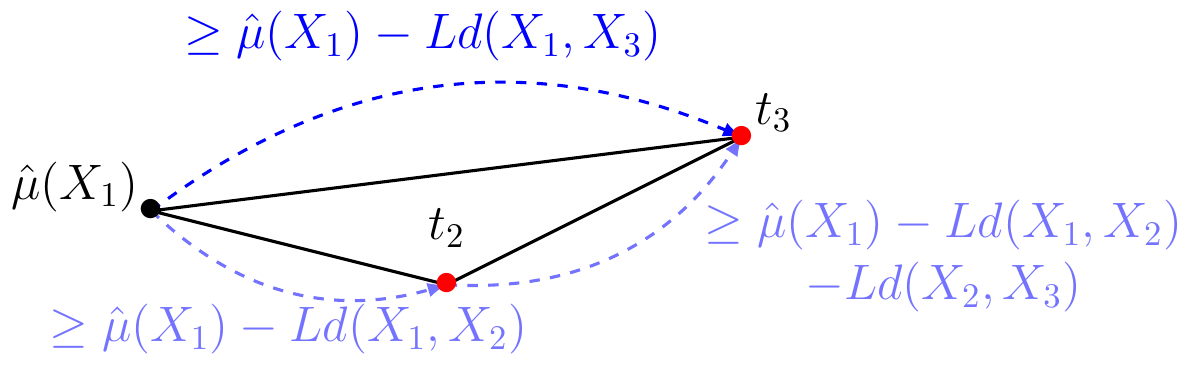}
	\caption{An example providing intuition for how, at the optimal solution of \eqref{eq:estim_lp}, constraints between pairs of points both in the no-overlap region are not active in the LP. The $i=1$ point (black) is in the overlap region, and the $i=2,3$ points (red) are in the no-overlap region. We must have $t_1=\hat{\mu}(X_1)$. The smoothness between 1 and 2 (and 1 and 3) imply bounds on $t_2$ (and $t_3$), with lower bounds shown along the blue arrows. Also shown is the lower bound implied on $t_3$ by $t_2$ when $t_2$ is set to $\hat{\mu}(X_1)-Ld(X_1, X_2)$ (which is the choice that will minimize the objective). Since $d(X_1, X_2)+d(X_2, X_3)>d(X_1, X_3)$, by the triangle inequality, the bound implied on $t_3$ by $t_2$'s bound from $t_1$ is always looser than the bound implied directly by $t_1$. Visually, the bound in dark blue is always tighter than the bound in light blue.}
	\label{fig:intuition}
\end{figure}

To see why this is the case, consider an example with $n=3$, where the $i=1$ point is in the overlap region and the other two points are in the no-overlap region. Since the objective is non-decreasing in the $t_i$, we would like $t_2$ to be as small as possible. We know that $t_2$ must satisfy $t_2\geq \hat{\mu}(X_1)-Ld(X_1, X_2)$, so we make it as small as possible by setting $t_2=\hat{\mu}(X_1)-Ld(X_2, X_1)$. Then, consider $t_3$. The lower bound on $t_3$ coming from $i=1$ is $\hat{\mu}(X_1)-Ld(X_3, X_1)$, while the lower bound coming from $i=2$ is $\hat{\mu}(X_1)-Ld(X_2, X_1)-Ld(X_3, X_2)$. The key point is that $d(X_2, X_1)+d(X_3,X_2)>d(X_3,X_1)$ by the triangle inequality, and so the bound from the overlap region is always sharper than the bound from the no-overlap region. A visual representation of this example appears in Figure~\ref{fig:intuition}. We emphasize that this no-interaction phenomenon is somewhat unusual and does not hold in general: for example, if we were to replace the Lipschitz smoothness assumption with an $\alpha$-H\"older continuity assumption, it would no longer be satisfied. 

Based on this intuition, we expect that constraints between points in the no-overlap region in \eqref{eq:estim_lp} are redundant, and the optimal solution to \eqref{eq:estim_lp} will simply set each $t_i$ in the no-overlap region to the tightest lower bound obtained from a point in the overlap region. We state this precisely in our next theorem, which requires the following assumptions.

\begin{assumption}
	\label{ass:lipschitz}
	The estimated $\hat{\mu}$ satisfies 
	\begin{equation}
		|\hat{\mu}(X_i)-\hat{\mu}(X_j)|\leq Ld(X_i, X_j), \text{ for all } X_i, X_j \text{ such that } \pi_b(X_i), \pi_b(X_j)>0.
	\end{equation}
\end{assumption}

\begin{assumption}
	\label{ass:consistency}
	The estimated $\hat{\mu}$ is consistent so that 
$
\sup_{x:\pi_b(x)>0} |\hat{\mu}(x)-\mu_{P_0}(x)|\xrightarrow{\P}0.
$
\end{assumption}
Assumption~\ref{ass:lipschitz} is necessary for \eqref{eq:estim_lp} to be feasible\textemdash otherwise we would be forced to have a pair $t_i, t_j$ in the overlap region with $|t_i-t_j|>Ld(X_i, X_j)$, and this pair would violate the Lipschitz constraints. To satisfy this assumption for moderate values of $L$, we recommend estimating $\hat{\mu}$ with smooth approximations such as parametric models, splines, or kernel methods \citep{hastie2009}. Tree-based methods, which give non-smooth approximations, may only satisfy Assumption~\ref{ass:lipschitz} for large values of $L$, unless a smoothed tree-based method such as that of \citet{friedberg2020} is used. Assumption~\ref{ass:consistency} requires consistency in the overlap region, so that $\hat{\mu}(X_i)$ is a good approximation of $\mu(X_i)$ for $X_i$ with $\pi_b(X_i)>0$. Note that Assumption~\ref{ass:consistency} makes no assumptions on the behavior of $\hat{\mu}$ in the no-overlap region.

Our next assumption is on the marginal distribution $P_{0,X}$ of $X_i$.
\begin{assumption}
	\label{ass:density}
	For every $x$ such that $\pi_b(x)>0$, either (a)~the distribution $P_{0,X}$ has an atom at $\{x\}$ or (b)~for every $\epsilon>0$, there exists $\delta>0$ such that $\P(d(X_i, x)\leq \epsilon)>\delta$.
\end{assumption}
This assumption ensures that the marginal distribution of $X_i$ does not have any ``holes'' that prevent us from observing parts of the overlap region. 

Under these assumptions, we have the following result, whose proof appears in Appendix~\ref{sec:proofs}

\begin{theorem}
	\label{thm:mu_hat_lip}
	Suppose that $P_0\in\mathcal{P}_L^{\Lip}$. Then, for $\hat{\psi}_2^-$ and $\psi_2^-$, we have that:
	\begin{enumerate}[(a)]
		\item under Assumption~\ref{ass:lipschitz}, the problem \eqref{eq:estim_lp} is feasible and has value 
			\begin{equation}
				\hat{\psi}_{\unid}^-=\frac{1}{n}\sum_{i=1}^n \pi_e(X_i)\left(\max_{j:\pi_b(X_j)>0} \hat{\mu}(X_j)-Ld(X_i, X_j)\right)\mathbf{1}\{\pi_b(X_i)=0\},
			\end{equation}
		\item the population bound is 
			\begin{equation}
				\psi_2^-= \E_{P_0}\left[ \pi_e(X_i)\left(\sup_{x: \pi_b(x)>0} \mu_{P_0}(x)-Ld(X_i, x)\right)\mathbf{1}\{\pi_b(X_i)=0\} \right],
			\end{equation}
		\item under Assumptions~\ref{ass:consistency} and \ref{ass:density}, we have $\hat{\psi}_2^-\xrightarrow{\P}\psi_2^-$.
	\end{enumerate}
\end{theorem}

The results of this theorem warrant further discussion. Theorem~\ref{thm:mu_hat_lip}(a) shows that we do not need to actually solve the problem \eqref{eq:estim_lp} numerically because a closed-form solution is available thanks to the no-interaction property. This is extremely important for even moderate values of $n$, because the problem \eqref{eq:estim_lp} has $O(n^2)$ constraints; since the fastest algorithms for solving a general linear program with $d$ constraints have complexity $O(d^{2.5})$ \citep{lee2015}, solving \eqref{eq:estim_lp} has worst-case complexity $O(n^5)$. In contrast, computing the closed-form given in \eqref{eq:closed_form} requires only $O(n^2)$ operations, which is of the same order as computing all of the pairwise distances $d(X_i, X_j)$.  So we see that the result of Theorem~\ref{thm:mu_hat_lip}(a) provides substantial efficiency gains. Even in settings where our results may not be directly applicable and \eqref{eq:estim_lp} does not have a closed-form (this may be the case if, for example, additional constraints based on domain knowledge are incorporated), being able to eliminate a subset of the constraints is highly valuable. Constraint elimination is in general LP-complete \citep{telgen1983identifying}, and so it is an appealing feature of \eqref{eq:estim_lp} that we can always eliminate a subset of the constraints. 

The results in Theorem~\ref{thm:mu_hat_lip}(b,c) are of a different flavor, and characterize the statistical properties of the method. They show that the solution to \eqref{eq:estim_lp} converges to a well-defined limit in large samples, and that this limit is exactly $\inf_{P\in \mathcal{P}_L^{\Lip}} \psi_{\unid}(P)$. The proof of (b) is essentially a continuous analogue of the proof of (a), and again relies on the no-interaction property. The main idea in the proof of (c) is that we can first replace $\hat{\mu}$ in $\hat{\psi}_2^-$ with $\mu$ by consistency, and then replace the maximum over $X_j$ in the overlap region with a supremum over all points in the overlap region when $n$ is sufficiently large. The result then follows from the law of large numbers. Together, (b) and (c) show that the condition of Theorem~\ref{thm:interval_consistency} is satisfied, and also that our bound is the best possible bound on $\psi(P_0)$ under the given assumptions. Even if we had infinite data, so that we fully knew the distribution $P_0$ in the overlap region, the most we could conclude is that $\psi_2(P_0)\geq {\psi}_{\unid}^{-}$.



We conclude by noting that, in cases where even the $O(n^2)$ operations required to compute $\hat{\psi}_{\unid}^-$ are too computationally expensive, we can construct conservative approximations of $\hat{\psi}_{\unid}^-$ that are computationally more effiecnt by building on recent advances in exact and approximate nearest neighbor search; we discuss these in Appendix~\ref{sec:nn}.

\subsection{Monotonicity assumptions}

Third, to highlight the versatility of our framework, we consider placing monotonicity assumptions on $\mu_{P_0}$. For example, in medical applications, if one covariate is a risk score, we may believe that the conditional mean of the outcome should be monotone with respect to that risk score. More generally for any ordering $\prec$ on $\mathcal{X}$, we can consider the family of distributions
\begin{equation}
	\mathcal{P}_{\prec}^{\text{mono}}=\left\{ P\text{ consistent w.~} P_0: \mu_P(x_1)\leq \mu_P(x_2)\text{ when }x_1\prec x_2\right\}.
\end{equation}
As in the previous section, we can approximate the infimum in \eqref{eq:inf_sup} by solving the optimization problem
\begin{mini}|s|
	{t_{1},\cdots, t_n}{\frac{1}{n}\sum_{i=1}^n t_i\pi_e(X_i)\mathbf{1}\{\pi_b(X_i)=0\}}
					{}{}
					\addConstraint{t_i\leq t_j,\quad 1\leq i, j\leq n \ \text{ s.t. }X_i\prec X_j}
					\addConstraint{t_i-\hat{\mu}(X_i)=0,\quad 1\leq i\leq n \ \text{ s.t. }\pi_b(X_i)>0}.
					\label{eq:oracle_lp_mono}
				\end{mini}
				Note that this linear program also satisfies the \emph{no-interaction property} we defined in Section~\ref{subsec:smooth}, since the ordering $\prec$ is transitive. As a result, if we set each $t_i$ in the no-overlap region to the minimum of its lower bounds from the overlap region, the constraints between the $t_i$ in the no-overlap region will be satisfied as well. Thus, the resulting estimator will enjoy properties like those in Theorem~\ref{thm:mu_hat_lip}.  We do not fully develop such results here because our main focus in the present work is on smoothness assumptions; the purpose of this example is to demonstrate that our framework encompasses other assumptions worthy of further study.



\subsection{Further assumptions and combinations of assumptions}



The framework we present captures many other potentially interesting assumptions. For example, we can combine the two assumptions presented here, and assume that $P\in \mathcal{P}_L^{\Lip}\cap \mathcal{P}_{\ell, u}^{\bdd}$, to obtain tighter bounds than either assumption alone would give. In fact, the no-interaction property discussed after \eqref{eq:estim_lp} continues to hold in this case, and our results then extend naturally; we present these more general results in Theorem~\ref{thm:mu_hat_lip_gen} of Appendix~\ref{sec:proofs}.

A variety of other assumptions are also possible, including monotonicity of $\mu$ with respect to a partial order on the covariates $X_i$, convexity of $\mu$, smoothness of higher derivatives of $\mu$, $\alpha$-H\"older continuity of $\mu$, and compositions of these assumptions. Interestingly, these assumptions and their compositions do not necessarily satisfy the no-interaction property: for example, if we assume both Lipschitz smoothness and monotonicity, the no-interaction property no longer holds, as we show in Appendix~\ref{sec:counter}. As such, our results on smoothness, boundedness, and their composition, are both non-trivial and surprising.



\section{Rates of convergence}
\label{sec:rates}

The results of the previous section show that we can construct an estimator $\hat{\psi}_2^-$ that is consistent for the sharp partial identification bound $\psi_2^{-}$ of $\psi_2(P_0)$ under the assumption that $P_0\in \mathcal{P}_L^{\Lip}$. In this section, we provide a more precise characterization of the asymptotics of $\hat{\psi}_2^-$ by identifying the rate at which it converges to $\psi_2^{-}$ through an upper bound, and showing that this rate is optimal up to log factors through a lower bound.

\subsection{Upper bounds}

Our first result is the following theorem, which bounds the mean-squared error (MSE) of $\hat{\psi}_2^-$.

\begin{theorem}
		Let $\hat{\psi}_2^-, \psi_2^{-}$ be as in Theorem~\ref{thm:mu_hat_lip}. Then, if $X_i$ has a density that is lower bounded by a constant $b$, we have the upper bound 
			\begin{equation}
				\E[(\hat{\psi}_2^- -\psi_2^{-})^2]\leq 2\E[\|(\hat{\mu}(x)-\mu(x))\mathbf{1}\{\pi_b(x)>0\}\|_{\infty}^2]+4L^2(c_dbn)^{-2/d}+\frac{\sigma^2}{n},
				\label{eq:mse_decomp}
			\end{equation}
			where $c_d$ is the volume of the unit ball in $d$ dimensions and 
\begin{equation}
	\sigma^2=\var\left( \pi_e(X_i)\left( \sup_{x:\pi_b(x)>0} \mu_{P_0}(x)-Ld(X_i, x) \right)\mathbf{1}\{\pi_b(X_i)=0\} \right).
			\end{equation}
\end{theorem}

The bound \eqref{eq:mse_decomp} contains three terms: the first term corresponds to the error incurred from the estimation of $\mu$ by $\hat{\mu}$; the second term corresponds to the error incurred from approximating the supremum over the overlap region by a maximum over the observed data points; the third term is the irreducible noise we would face even if we had full knowledge of $\mu$ and could compute the supremum. 

Classical results from the literature on non-parametric regression show that the minimax rate for estimating an $L$-Lipschitz function in $d$ dimensions is of order $(\log n/n)^{-2/(d+2)}$, and that this rate is obtained by a kernel estimator \citep{stone1982}. Thus, the dominant term of \eqref{eq:mse_decomp} is the first one, and the MSE of $\hat{\psi}_2^-$ is $O( (\log n/n)^{-2/(d+2)})$ when using an appropriate kernel estimator.

\subsection{Lower bounds}

Next we present a nearly matching lower bound, proven in Appendix~\ref{sec:proofs} through a LeCam two-point argument that extends lower bound constructions for estimating Lipschitz functions at a point to the off-policy evaluation setting \citep{wainwright2019}.
In the set-up of this theorem, we make an assumption on the support of $\pi_b$ to simplify the lower bound construction; we expect that the same rate holds regardless of the geometry of the support of $\pi_b$.

\begin{theorem}
	\label{thm:lower_bound}
	Let $\psi_2^{-}$ be as in Theorem~\ref{thm:mu_hat_lip}, suppose that the covariates $X_i$ takes values in $[-1,1]^d$, and that the support of the policy $\pi_b$ is $[-1/2, 1/2]^d$. Then, if $n\geq 2^{1-d}L^2$, we have the minimax lower bound
	\begin{equation}
		\inf_{\hat{\psi}_2^-}\sup_{P \in \mathcal{P}_L^{\Lip}}\E_P\left[ (\hat{\psi}_2^--\psi_2^{-})^2 \right]\geq \frac{1}{16}\E[\pi_e(X_i)\mathbf{1}\{\pi_b(X_i)=0\}]^2(2n)^{-2/(d+2)}(4L)^{-2d/(d+2)}.
	\end{equation}
\end{theorem}

The lower bound of Theorem~\ref{thm:lower_bound} shows that any estimator must have an MSE that is at least of order $n^{-2/(d+2)}$, showing that the rate achieved by $\hat{\psi}_2^-$ is optimal up to log-factors. To provide further insight on why we obtain an $n^{-2/(d+2)}$ rate in this problem, note that the estimand $\psi_2^{-}$ depends on a supremum of over $\mu$ the overlap region. Thus, we essentially need to be able to estimate the value of $\mu$ at a particular point, namely the point at which the supremum is attained. The minimax lower bound for estimating an $L$-Lipschitz function in $d$ dimensions at a point is $n^{-2/(d+2)}$, and so we inherit this rate for this estimation of $\psi_2^{-}$. 

Based on this intuition, we expect that similar lower bounds will continue to hold for for any non-parametric assumption that requires estimation of $\mu$ in the overlap region. For example, estimating partial identification bounds under $\mathcal{P}_{\prec}^{\text{mono}}$ will also require estimation of $\mu$ in the overlap region, and thus will have a similar $n^{-2/(d+2)}$ rate. In contrast, estimation of partial identification bounds under the boundedness assumption of $\mathcal{P}_{\ell, u}^{\bdd}$ does not require estimating $\mu$ in the overlap region, and thus can be done at the fast $n^{-1}$ rate. 

These distinctions then lead to differences in inferential procedures: inference for the partial identification bounds under a boundedness assumption can be done using the confidence intervals of \citet{imbens2004}, but performing similar inferences under our non-parametric smoothness assumptions is more challenging. The slow rate of convergence precludes the possibility of obtaining a central limit theorem for $\hat{\psi}_2^-$ at the $\sqrt{n}$-rate, suggesting that standard methods for constructing confidence intervals are not applicable, and that the construction of such confidence intervals is an exciting direction for future work.


\section{Experiments}
\label{sec:experiments}

We now demonstrate our methods in two semi-synthetic settings. In the first setting, we study the coverage guarantees of Theorem~\ref{thm:interval_consistency} and compare the width of our intervals to the width of the intervals of \citet{ben2021}; in the second, we demonstrate the utility of our methods in a real-world setting. For another example of how our methods perform on real data, we refer interested readers to \citet{saveski2023}, which considers the partial identification approaches based on smoothness and monotonicity presented in this work when evaluating policies in peer review management systems. 

\begin{table}
	\centering
	\resizebox{\linewidth}{!}{
	\begin{tabular}{lllllll}
		\toprule
		$L$&$n=1000$&$n=2000$&$n=3000$&$n=4000$&$n=5000$&$n=10000$\\
		\midrule
		$1$&0.408 (0.785)&0.0003 (0.002)&0.0 (0.00)&0.0 (0.00)&0 (0.00)&0 (0.00)\\
		$2$&0.597 (1.00)&0.725 (1.00)&0.789 (1.00)&0.827 (1.00)&0.859 (0.995)&0.095 (0.106)\\
		$3$&0.646 (1.00)&0.772 (1.00)&0.828 (1.00)&0.869 (1.00)&0.897 (1.00)&0.969 (1.00)\\
		$4$&0.682 (1.00)&0.802 (1.00)&0.853 (1.00)&0.891 (1.00)&0.919 (1.00)&0.980 (1.00)\\
		$5$&0.706 (1.00)&0.820 (1.00)&0.868 (1.00)&0.907 (1.00)&0.931 (1.00)&0.984 (1.00)\\
		$\infty$&0.750 (1.00)&0.861 (1.00)&0.907 (1.00)&0.935 (1.00)&0.956 (1.00)&0.992 (1.00)\\
		\bottomrule
	\end{tabular}
	}

	\caption{Coverage (as defined in Theorem~\ref{thm:interval_consistency} with $\epsilon=0.01$) of partial identification intervals for the value of $\pi_e$ on the yeast dataset at a range of samples sizes $n$ and smoothness parameters $L$. Also shown in parentheses are the rates at which Assumption~\ref{ass:lipschitz} is satisfied. All results are averaged over 10,000 replications. We see that the Lipschitz assumption with $L=1,2$ likely does not hold, since Assumption~\ref{ass:lipschitz} is not satisfied in larger sample sizes, while the Lipschitz assumption for larger values of $L$ seems to hold, since Assumption~\ref{ass:lipschitz} is satisfied and coverages approach the desired 100\% rate. \vspace*{0.1in}}
	\label{tab:yeast}
\end{table}

\subsection{Yeast dataset}

\paragraph{Coverage analysis.}
Following prior work on off-policy evaluation, we conduct a semi-synthetic experiment by converting a classic multi-class classification dataset into an off-policy evaluation dataset \citep{dudik2011, wang2017, wu2018, su2020, zhan2021}. Like those prior works, we use the yeast data set from the UCI repository \citep{Dua:2019}, which consists of $n=1,484$ observations of data points $(X_i, \tilde{Y}_i)$, where $X_i$ is a covariate vector of length $d=8$ all of whose entries lie in $[0,1]$ and $\tilde{Y}_i$ is a class label indicating one of 10 classes. Of these 10 classes, 6 are quite rare, so we remove them from the dataset to ensure that all classes will be represented when resampling from the data. This leaves $n=1,299$ observations, each of which is labelled with one of 4 classes. 

We treat the empirical distribution of the observed data as the true data-generating distribution $P_0$, and sample with replacement from the observed data to generate samples of different sizes, $n$. The action set $\mathcal{A}$ is the set of 4 classes represented in the data, and the evaluation policy, $\pi_e$, samples actions from the fitted probabilities of a logistic regression. The behavior policy, $\pi_b$, samples actions from the fitted probabilities of the same logistic regression, but with probabilities below 0.05 set to zero. This configuration leads to an overlap violation between the behavior and evaluation policies. The outcomes $Y_i$ are 1 when the action taken matches the true class label $\tilde{Y}_i$, and 0 otherwise. 

We aim to study the guarantees of Theorem~\ref{thm:interval_consistency} and the role of the smoothness parameter $L$. Using the methods of Section~\ref{sec:estimator}, we can generate partial identification intervals for the value of $\pi_e$ under the assumption that $P_0\in \mathcal{P}_{L}^{\Lip}\cap \mathcal{P}_{0, 1}^{\bdd}$, where the Lipschitz assumption is made with respect to Euclidean distance in the covariate space. We fit $\hat{\mu}$ using a logistic regression with the default regularization of scikit-learn \citep{scikit-learn}. Table~\ref{tab:yeast} shows the rates at which these intervals cover the true off-policy value, in the sense of Theorem~\ref{thm:interval_consistency} with $\epsilon=0.01$, for a range of values of $L$. Also shown in parentheses are the fraction of times that $\hat{\mu}$ is $L$-Lipschitz in the overlap region, i.e., the fraction of times that Assumption~\ref{ass:lipschitz} is satisfied. When Assumption~\ref{ass:lipschitz} is not satisfied, the partial identification bounds are undefined, and so we never cover the off-policy value. Note that $L=\infty$ corresponds to pure Manski bounding, or equivalently to assuming only that $P_0\in \mathcal{P}_{0,1}^{\bdd}$, an assumption that is satisfied by construction in this example. We provide a plot of the results of Table~\ref{tab:yeast}, as well as results with $\epsilon=0.005$, in Appendix~\ref{sec:app_sims}.

In Table~\ref{tab:yeast}, we see a sharp distinction between $L=1, 2$ and values of $L$ greater than 2. For the small values of $L$, the problem gradually becomes ceases to become feasible for larger values of $n$, and so the resulting intervals are undefined and rarely cover the off-policy value. For larger values of $L$, the problem is always feasible at all values of $n$, and the coverage increases with $n$, approaching the desired 100\% coverage in large sample sizes. In particular, the fact that the coverage of intervals constructed assuming, for example, that $L=4$ or $L=5$ (an assumption that is not satisfied by construction) is close to the coverage of intervals satisfied by the $L=\infty$ boundedness assumption (which is satisfied by construction), suggests that these smoothness assumption for these larger values of $L$ is quite plausible. As a point of reference, the $L=5$ assumption on these covariates in the Euclidean metric means, for example, that if two units agree in all but one covariate, and disagree by 0.2 in that covariate, their expected outcome can differ by 1, so they do not place any bounds on each other. 


\paragraph{Interval width analysis.}
We now study the width of our intervals, using the intervals proposed by \citet{ben2021} as a baseline, on the Yeast data. This baseline method constructs lower bounds on the off-policy value using a simultaneous confidence interval for $\hat{\mu}$, rather than the fitted $\hat{\mu}$ itself. We use the same dataset and construct the behavior and evaluation policies in the same way as in the previous section, but discretize the original covariates $X_i$ according to the map $X_i\mapsto \mathbf{1}\{X_i<0.5\}$.
We discretize since, as the authors descibe, constructing the simultaneous confidence interval needed by the baseline method is challenging in the case of continuous covariates.

Table~\ref{tab:bm_compare} shows the results for $L=1$ and a range of sample sizes. We find that the intervals obtained using the baseline methods are, in this setting, consistently 40-50\% wider than those obtained using the proposed method. This conservativeness is induced by the simultaneous confidence interval construction, which the baseline method relies on to obtain better theoretical guarantees for policy learning, but is needlessly loose for policy evaluation.

\begin{table}
	\centering
	\begin{tabular}{lllllll}
		\toprule
		&\multicolumn{5}{c}{Sample size, $n$}\\
		\cmidrule{2-7}
		Method &500&1000&1500&2000&2500&3000\\
		\midrule
Proposed&4.37&3.09&2.59&2.15&1.80&1.68\\
Baseline&6.57&4.89&3.89&3.18&2.64&2.41\\
	\midrule
Ratio&1.50&1.58&1.50&1.48&1.47&1.43\\
		\bottomrule
	\end{tabular}
	\caption{
	Average width of the intervals obtained using the proposed method in semi-synthetic experiments on the Yeast dataset.
	We compare to the intervals of a baseline method proposed by \citet{ben2021} and compute the ratio between the width of the two intervals.
	We set $L=1$, average over 1000 trials, and scale the results by $10^3$ for readability.
	We find that the intervals of baseline method are 40-50\% wider than the proposed method.
}

	\label{tab:bm_compare}
\end{table}

\subsection{Yahoo!\ Front Page Today dataset.} Our second experiment uses the Yahoo Webscope's featured news dataset, a standard benchmark for OPE algorithms \citep{yahoo, li2010, li2011}.

This dataset was collected over 10 days in May 2009, and consists of observations of user visits and actions on the front page of the Yahoo!\ web portal. Each observation is a tuple $(X_i, A_i, Y_i)$, where $X_i$ is a 5-dimensional vector of covariates, all of which lie in $[0,1]$, $A_i$ is an article shown to the user in a featured position on the page, and $Y_i$ is a binary response indicating whether or not the user clicked. Additionally, each $A_i$ is accompanied by a 5-dimensional covariate vector $V_i$. The articles $A_i$ are chosen from a hand-curated pool of articles that is updated each hour. We restrict our focus to a single hour of data so that the articles are drawn from a fixed pool $\mathcal{A}$ of articles, considering $n=16,628$ data points. Then, this problem has the form of a multi-armed bandit problem, as described in Section~\ref{sec:model}, with action space $\mathcal{A}$.

The architects of this dataset sampled $A_i\sim \text{Unif}(\mathcal{A})$. Uniform sampling guarantees overlap, but requires taking many sub-optimal actions. A tuned, non-uniform logging policy would clearly be preferable. With our methods, we show that it is possible to run a non-uniform policy and yet still evaluate other policies that do not satisfy overlap.

Specifically, we consider the following scenario: from historical data, we know that the user covariate $X_{i,3}$ is positively correlated with the response $Y_i$ and that the article covariate $V_{i,0}$ is positively correlated with $Y_i$. Based on this knowledge, we believe that we should avoid showing articles with low values of $V_{i,0}$ to users with low values of $X_{i,3}$. This leads us to consider the family of policies 
\begin{equation}
	\pi^{(T)}(X_i, a)=\left\{ \begin{aligned} &1/|\mathcal{A}|&&\text{ if }&&& X_{i,3}>T,\\
	&\mathbf{1}\{a\in \mathcal{A}^*\}/|\mathcal{A}^*|&&\text{ if }&&&X_{i,3}\leq T,\end{aligned}\right.
	\label{eq:pi_T}
\end{equation}
where $\mathcal{A}^*$ is a subset of articles we expect to perform well and $T$ is a cut-off that identifies users who are unlikely to click. Here, we take $\mathcal{A}^*$ to be the set of articles with above median values of $V_{i,0}$.

We suppose that we have deployed the policy $\pi^{(T)}$ with $T=0.5$ (and simulate having done so by subsampling), and are interested in exploring other values of the threshold $T$. The subsampled dataset contains $n=10,086$ data points. For values of $T \ge 0.5$, the behavior policy provides full support for the evaluation policy, and no partial identification is required. For values of $T<0.5$, there is an overlap violation, since $\pi^{(T)}$ can show users with $X_{i,3}\in (T, 0.5)$ articles $a\in \mathcal{A}\setminus \mathcal{A}^*$, an action that $\pi^{(0.5)}$ assigns zero probability. 

\begin{figure}[t]
	\centering
	\includegraphics[width=\textwidth]{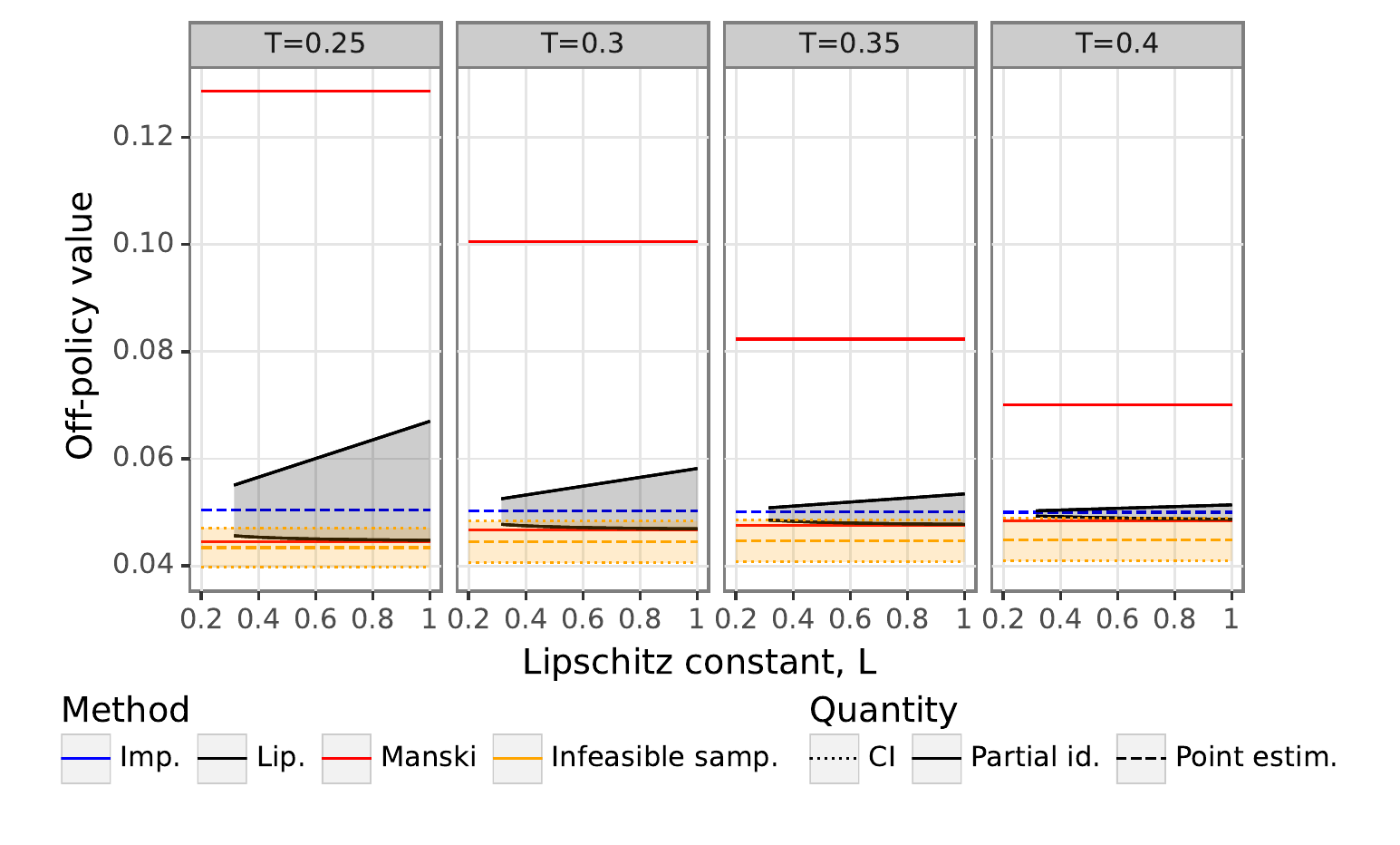}
	\caption{Partial identification bounds under the assumption that $P_0\in \mathcal{P}_L^{\Lip}\cap \mathcal{P}_{0,1}^{\bdd}$ (black), Manski bounds under the assumption that $P_0\in \mathcal{P}_{0, 1}^{\bdd}$ (grey), and pure imputation estimates (blue) of the value of the policy $\pi^{(T)}$ estimated using data from the behavior policy $\pi^{(0.5)}$ for a range of $T$ and $L$. Also shown are the point estimate and confidence intervals estimated using infeasible sample data from the uniform behavior policy. We see that the pure imputation estimate overestimates the value of $\pi^{(T)}$, but that our partial identification bounds correct for this. Crucially, the width of our partial identification interval increases as the model estimate and infeasible sample estimate diverge, meaning that we are correctly adjusting for the model's ability to extrapolate into the no-overlap region. For reference, there are overlap violations for 16.8\%, 10.8\%, 7.0\%, and 4.3\% of the points when $T=0.25$, $T=0.3$, $T=0.35$, and $T=0.4$ respectively.}
	\label{fig:yahoo_ope_manski}
\end{figure}

Since we have overlap violations, we use the infimum estimator of Theorem~\ref{thm:mu_hat_lip} (and its analogue for the supremum) to construct partial identification intervals $[\hat{\psi}^-, \hat{\psi}^+]$ as in Theorem~\ref{thm:interval_consistency} under the assumption that $P_0\in \mathcal{P}_L^{\Lip}\cap \mathcal{P}_{0,1}^{\bdd}$, where the Lipschitz assumption is again made with respect to Euclidean distance in the covariate space. Our results are shown in Figure~\ref{fig:yahoo_ope_manski}, which plots the interval estimators $[\hat{\psi}^-, \hat{\psi}^+]$ of the value of $\pi^{(T)}$ under Lipschitz assumptions for $T=0.25, 0.3,0.35, 0.5$, with $\hat{\mu}$ fit using a logistic regression with the default regularization of \cite{scikit-learn} (which leads to smooth $\hat{\mu}$), for a range of Lipschitz parameters $L$. For results with larger values of $L$ and $T$, we refer interested readers to Appendix~\ref{sec:app_sims}. As a point of reference, the $L=1$ assumption on this data means, for example, that if two units agree in all but one covariate, and differ by the maximum possible value of 1 in that covariate, their expected outcome can differ by the maximum possible value of 1, so they place no bounds on each other. 

Also shown are the point estimates obtained by using model predictions without any partial identification, as well as an infeasible sample estimate (and accompanying confidence intervals) of the value of $\pi^{(T)}$ as estimated on data from a uniform policy. This infeasible sample estimate corresponds to the results of the experiment we would run (at considerable cost) if unwilling to rely on smoothness assumptions. 

The key takeaway from Figure~\ref{fig:yahoo_ope_manski} is that, as $T$ increases, and the imputation estimate and infeasible sample estimate grow closer, the width of our partial identification interval correspondingly decreases. This is because, as $T$ becomes smaller, the amount of units in the no-overlap region $\{X_i: T\leq X_{i,3}\leq 0.5\}$ increases, and the maximum distance between points in the overlap region and points in the no-overlap region increases. Our method accounts for this, and correctly distinguishes cases where the model needs to be corrected only slightly from cases where it must be corrected substantially. Furthermore, for larger values of $L$, our intervals consistently overlap with the 95\% confidence interval from the alternative (``costly'') experiment under a uniform policy. 

We can also compare against the Manski partial identification regions obtained solely from the assumption that $P\in \mathcal{P}_{0,1}^{\text{bdd}}$, shown in gray. For instance, one may be concerned that the intersection between our intervals and the 95\% confidence interval is not large, but because this is true for the Manski intervals as well (whose non-parametric assumption holds exactly), we conclude that the small intersection is due to error in the estimation of $\hat{\psi}_1$. More importantly, we see that our intervals are much narrower than the Manski intervals for small values of $L$ and converge to them as $L\to \infty$. For example, when $L=1$, our intervals are 73.5\%, 79.1\%, 83.7\%, and 91.5\% narrower than the Manksi intervals for $T=0.25$, $T=0.3$, $T=0.35$, and $T=0.4$ respectively.

\section{Discussion}


In this paper we have developed partial identification results for off-policy evaluation under smoothness assumptions, shown that our bounds have good asymptotic properties, and demonstrated their value on a real-data example. 

When using our methods, we do not recommend pre-committing to a particular value of $L$. Rather, we recommend defining criteria a policy should meet and evaluating the plausibility of the smoothness assumptions necessary to ensure that those criteria are met. For example, in the setting of Figure~\ref{fig:yahoo_ope_manski}, if we were interested in deploying the policy $\pi^{(0.3)}$, but only wanted to deploy it if it had at least a 4.5\% click rate, we would be able to ensure this as long as $L$ is at most $0.5$. The plausibility of the $L=0.5$ assumption can then be verified based on historical data and domain knowledge. 

In cases where we would like to report a single interval estimate of the off-policy value, we must choose a particular value of $L$. We can lower bound the value of $L$ by the maximum value of $|\hat{\mu}(X_i)-\hat{\mu}(X_j)|/d(X_i, X_j)$ in the overlap region. This lower bound comes with two caveats: first, it is sensitive to the choice of $\hat{\mu}$, and we recommend fitting $\hat{\mu}$ using smooth methods to obtain moderate values of $L$. Second, this lower bound is, as advertised, only a lower bound\textemdash the conditional outcome may be much less smooth in the no-overlap region, making the true value of $L$ much larger. Choosing between values of $L$ above this lower bound requires a subjective judgement about the structure of the problem, and can only be justified by problem-specific considerations. 

Finally, since our methods are tied to a particular choice of covariate space and metric, it is crucial to consider whether or not the smoothness assumption is plausible and useful in that space and metric. For example, in an extremely high-dimensional space, most points may be quite far from each other, and so bounds obtained from the Lipschitz assumption will be quite loose. Similarly, in a categorical covariate space with a discrete metric like the Hamming distance, the expected outcome may only be smooth with large values of $L$, since the mean response in one category may be very different from the mean response in another. The tighter bounds we give are only as useful as the underlying assumptions are plausible, and so it is essential to combine our methodological proposals with thoughtful application. 

\paragraph{Acknowledgements} This work was supported in part by NSF CAREER Award \#2143176. We are grateful to Kevin Guo, Tal Wagner, Nihar Shah, and Steven Jecmen for helpful discussions. 

\bibliographystyle{plainnat}
\bibliography{lipschitz_ope}

\appendix
\section{Proofs of results}
\label{sec:proofs}

\subsection{Proof of Theorem~\ref{thm:interval_consistency}}

\begin{proof}[Proof of Lemma~\ref{thm:interval_consistency}]
	We begin by writing 
	\begin{align}
		\P(\hat{\psi}^--\epsilon\leq \psi(P_0)\leq \hat{\psi}^++\epsilon)&= 1-\P(\psi(P_0)<\hat{\psi}^--\epsilon\text{ or }\psi(P_0)>\hat{\psi}^++\epsilon),\\
		&\geq 1-\P(\psi(P_0)\leq \hat{\psi}^--\epsilon)-\P(\psi(P_0)>\hat{\psi}^++\epsilon),\label{eq:union_bound}
	\end{align}
	by the union bound. We now show that $\P(\psi(P_0)>\hat{\psi}^++\epsilon)\to 0$. An analogous argument shows that $\P(\psi(P_0)<\hat{\psi}^--\epsilon)\to 0$ as well, and these two facts along with \eqref{eq:union_bound} imply the result. 

	We have 
	\begin{align}
		\P(\psi(P_0)>\hat{\psi}^++\epsilon)&= \P(\psi_1(P_0)+\psi_2(P_0)>\hat{\psi}_1+\hat{\psi}_2^++\epsilon),\\
		&\leq \P(\psi_1(P_0)>\hat{\psi}_1+\epsilon/2)+\P(\psi_2(P_0)>\hat{\psi}_2^++\epsilon/2),\\
		&\leq \P(\psi_1(P_0)>\hat{\psi}_1+\epsilon/2)+\P\left(\sup_{P\in\mathcal{P}} \psi_2(P)>\hat{\psi}_2^++\epsilon/2\right),\\
		&\leq \P(|\psi_1(P_0)-\hat{\psi}_1|>\epsilon/2)+\P\left( \left|\sup_{P\in\mathcal{P}}\psi_2(P)-\hat{\psi}_2^+\right|>\epsilon/2 \right).\label{eq:consistency_bounds}
	\end{align}
	Both terms of \eqref{eq:consistency_bounds} are $o(1)$ under the given consistency assumptions, so we conclude that $\P(\psi(P_0)>\hat{\psi}^++\epsilon)=o(1)$ as desired.
\end{proof}

\subsection{Proof of Theorem~\ref{thm:mu_hat_lip}}
As discussed in the main text, we in fact prove a generalization of Theorem~\ref{thm:mu_hat_lip} that is based on the assumption that $P_0\in \mathcal{P}_L^{\Lip}\cap \mathcal{P}_{\ell, u}^{\text{bdd}}$. In this setting, we define $\hat{\psi}_2^-$ as the solution to the following optimization problem:
\begin{mini}|s|
	{t_1,\cdots, t_n}{\frac{1}{n}\sum_{i=1}^n t_i\pi_e(X_i)\mathbf{1}\{\pi_b(X_i)=0\}}
					{}{}
					\addConstraint{|t_i-t_j|\leq Ld(X_i, X_j),\quad 1\leq i<j\leq n}
					\addConstraint{t_i-\hat{\mu}(X_i)=0,\quad 1\leq i\leq n\text{ s.t. }\pi_b(X_i)>0}
					\addConstraint{\ell \leq t_i\leq u,\quad 1\leq i\leq n}
					\label{eq:estim_lp_gen}
				\end{mini}
				Then, we have the following result.
\begin{theorem}
	\label{thm:mu_hat_lip_gen}
	Suppose that $P_0\in \mathcal{P}_L^{\Lip}\cap \mathcal{P}_{\ell,u}^{\bdd}$. Then, for $\hat{\psi}_2^-$ and $\psi_2^-$, we have that:
	\begin{enumerate}[(a)]
		\item under Assumption~\ref{ass:lipschitz}, the problem \eqref{eq:estim_lp} is feasible and has value 
			\begin{equation}
				\hat{\psi}_2^-=\frac{1}{n}\sum_{i=1}^n \pi_e(X_i)\left( \ell \vee \max_{j:\pi_b(X_j)>0} \hat{\mu}(X_j)-Ld(X_i, X_j)\right)\mathbf{1}\{\pi_b(X_i)=0\};
				\label{eq:closed_form}
			\end{equation}
		\item the population bound is  
			\begin{equation}
				\psi_2^{-}=\E\left[ \pi_e(X_i)\left(\ell \vee \sup_{x: \pi_b(x)>0} \mu_{P_0}(x)-Ld(X_i, x)\right)\mathbf{1}\{\pi_b(X_i)=0\} \right];
				\label{eq:manski_sup}
			\end{equation}
		\item under Assumptions~\ref{ass:consistency} and \ref{ass:density}, we have $\hat{\psi}_2^-\xrightarrow{\P}\psi_2^-$
	\end{enumerate}
\end{theorem}

Theorem~\ref{thm:mu_hat_lip} from the main text follows from sending $\ell \to -\infty$ in Theorem~\ref{thm:mu_hat_lip_gen}. We now prove each part of the theorem in turn. In our proofs, we assume for the sake of convenience that the supremum in \eqref{eq:manski_sup} is attained. This will be the case if, for instance, $\{x: \pi_b(x)>0\}$ is compact. If the supremum is not attained, slight modifications of our proofs can be used to obtain the result.

\subsubsection{Proof of Theorem~\ref{thm:mu_hat_lip_gen}(a)}

\begin{proof}[Proof of Theorem~\ref{thm:mu_hat_lip_gen}(a)]
	To show the result, we must do two things: show that the objective of \eqref{eq:estim_lp} is bounded below by $\hat{\psi}_2^-$, and show that this bound is attained.

For the bound, observe that, for each $t_i$, we have 
\begin{align}
	t_i&\geq t_j-Ld(X_i, X_j)\text{ for }1\leq j\leq n,\\
	&\geq \max_{j:\pi_b(X_j)>0} t_j-Ld(X_i, X_j),\\
	&\geq \max_{j: \pi_b(X_j)>0} t_j-Ld(X_i, X_j),\\
	&= \max_{j: \pi_b(X_j)>0}\hat{\mu}(X_j)-Ld(X_i, X_j),
	\label{<+label+>}
\end{align}
where the first bound is the Lipschitz constraint and the last equality is the equality constraint. We also have $t_i\geq \ell$ for all $i$, and thus we must have
\begin{equation}
	\sum_{i=1}^n t_i\pi_e(X_i)\mathbf{1}\{\pi_b(X_i)=0\}\geq \sum_{i=1}^n \pi_e(X_i)\left(\ell \vee \max_{j: \pi_b(X_j)>0} \hat{\mu}(X_j)-Ld(X_i, X_j)  \right)\mathbf{1}\{\pi_b(X_i)=0\}=\hat{\psi}_2^-,
\end{equation}
as desired.

Next we will show that that there exist a set of feasible $t_i$ for which the objective of \eqref{eq:estim_lp} is equal to $\hat{\psi}_2^-$. The construction is 
\begin{equation}
	t_i^*= \ell\vee \max_{j:\pi_b(X_j)>0} \hat{\mu}(X_j)-Ld(X_i, X_j)
	\label{eq:t_star}
\end{equation}
This construction clearly has objective value $\hat{\psi}_2^-$, so we need only check that it is feasible. 

For the equality constraint, note that if $\pi_b(X_i)>0$, then the maximum in \eqref{eq:t_star} includes $j=i$, and thus is at least $\hat{\mu}(X_i)$. But for any $j\neq i$, we have $|\hat{\mu}(X_j)-\hat{\mu}(X_i)|\leq Ld(X_i,X_j)$ by the feasibility of \eqref{eq:estim_lp}, and so $\hat{\mu}(X_i)\geq \hat{\mu}(X_j)-Ld(X_i, X_j)$. Since $\hat{\mu}$ is range-bounded, we also have $\hat{\mu}(X_i)\geq \ell$, and so 
\begin{equation}
	\ell\vee \max_{j:\pi_b(X_j)>0} \hat{\mu}(X_j)-Ld(X_i, X_j)=\hat{\mu}(X_i)
\end{equation}
for $i$ in the overlap region, and the equality constraint is satisfied.

Next we check the Lipschitz constraint. To do this, we define 
\begin{equation}
	f(i)=\argmin_{j: \pi_b(X_j)>0} \hat{\mu}(X_j)-Ld(X_i, X_j)
\end{equation}
so that $t_i^*=\ell\vee (\hat{\mu}(X_{f(i)})-Ld(X_i, X_{f(i)}))$. We consider any two $t_i^*, t_j^*$, and distinguish three cases. 

First, if $\hat{\mu}(X_{f(i)})-Ld(X_i, X_{f(i)})<\ell$ and $\hat{\mu}(X_{f(j)})-Ld(X_j, X_{f(j)})<\ell$, then $|t_i^*-t_j^*|=0$ and the Lipschitz condition is satisfied.

Second, if $\hat{\mu}(X_{f(i)})-Ld(X_i, X_{f(i)})>\ell$ and $\hat{\mu}(X_{f(j)})-Ld(X_j, X_{f(j)})>\ell$, we may assume without the loss of generality that $t_i^*>t_j^*$, and compute
\begin{align}
	|t_i^*-t_j^*|&= t_i^*-t_j^*,\\
&=\left(\hat{\mu}(X_{f(i)})-Ld(X_i, X_{f(i)})\right)-\left(\hat{\mu}(X_{f(j)})-Ld(X_j, X_{f(j)})\right),\\
&\leq \left(\hat{\mu}(X_{f(i)})-Ld(X_i, X_{f(i)})\right)-\left(\hat{\mu}(X_{f(i)})-Ld(X_j, X_{f(i)})\right),\\
&= L\left( d(X_j, X_{f(i)})-d(X_i, X_{f(i)}) \right),\\
	&\leq Ld(X_i, X_j),
	\label{eq:no_interaction}
\end{align}
where the first inequality uses the maximality (by definition) of $f(j)$, and the second inequality is the triangle inequality $d(X_j, X_{f(i)})\leq d(X_i, X_{f(i)})+d(X_i, X_j)$. Thus the Lipschitz constraint is satisfied in this case. 

Finally, if $\hat{\mu}(X_{f(i)})-Ld(X_i, X_{f(i)})>\ell$ and $\hat{\mu}(X_{f(j)})-Ld(X_j, X_{f(j)})<\ell$, then 
\begin{align}
	|t_i^*-t_j^*|&= \hat{\mu}(X_{f(i)})-Ld(X_i,X_{f(i)})-(\ell),\\
	&\leq \hat{\mu}(X_{f(i)})-Ld(X_i,X_{f(i)})-(\hat{\mu}(X_{f(j)})-Ld(X_j, X_{f(j)}))\label{eq:manski_lip}
\end{align}
and \eqref{eq:manski_lip} is bounded by $Ld(X_i, X_j)$ by the arguments of \eqref{eq:no_interaction}.

Thus we see that the $t_i^*$ defined in \eqref{eq:t_star} are feasible, and conclude that the value of \eqref{eq:estim_lp} is $\hat{\psi}_2^-$.
\end{proof}

\subsubsection{Proof of Theorem~\ref{thm:mu_hat_lip_gen}(b)}

The proof of Theorem~\ref{thm:mu_hat_lip_gen}(b) is essentially a continuous version of the proof of Theorem~\ref{thm:mu_hat_lip_gen}(a).

\begin{proof}[Proof of Theorem~\ref{thm:mu_hat_lip_gen}(b)]
	We would like to show that 
	\begin{equation}
		\inf_{P\in \mathcal{P}_L^{\Lip}} \E_P\left[ Y_i\pi_e(X_i)\mathbf{1}\{\pi_b(X_i)=0\} \right]=\E_{P_0}\left[ \pi_e(X_i)\left( \ell\vee \sup_{x:\pi_b(x)>0} \mu_{P_0}(x)-Ld(X_i,x) \right)\mathbf{1}\{\pi_b(X_i)=0\} \right].
		\label{eq:wts}
	\end{equation}
	To do this, we must first show that each for each $P\in\mathcal{P}_L^{\Lip}\cap \mathcal{P}_{\ell, u}^{\bdd}$, $\psi_2(P)$ is greater than $\psi_2^{-}$, and then show that $\psi_2^{-}$ is attained.

	For the lower bound, note that for any $X_i$ and $x$ such that $\pi_b(X_i)=0$ and $\pi_b(x)>0$, we must have $\mu_P(X_i)\geq \mu_P(x)-Ld(X_i, x)$ since $\mu_P$ is $L$-Lipschitz. Furthermore, since $P$ is consistent with $P_0$ for $x$ such that $\pi_b(x)>0$, we in fact have $\mu_P(X_i)\geq \mu_{P_0}(x)-Ld(X_i,x)$ for all such $x$. Lastly, since $\mu_P$ is bounded, we also have $\mu_P(X_i)\geq \ell$. Thus, 
	\begin{align}
		\E_P[Y_i\pi_e(X_i)\mathbf{1}\{\pi_b(X_i)=0\}]&= \E_P[\pi_e(X_i)\mu_P(X_i)\mathbf{1}\{\pi_b(X_i)=0\}],\\
		&\geq \E_P\left[ \pi_e(X_i)\left( \ell\vee \sup_{x:\pi_b(x)>0} \mu_{P_0}(x)-Ld(X_i, x) \right)\mathbf{1}\{\pi_b(X_i)=0\}\right],\\
		&= \E_{P_0}\left[ \pi_e(X_i)\left( \ell\vee \sup_{x:\pi_b(x)>0} \mu_{P_0}(x)-Ld(X_i, x) \right)\mathbf{1}\{\pi_b(X_i)=0\}\right],\label{eq:lb}
	\end{align}
	where the first equality is the tower rule, the second inequality follows from the arguments of the preceeding paragraph, and the third equality uses the fact that $P$ is consistent with $P_0$. Since \eqref{eq:lb} is exactly $\psi_2^{-}$, this shows the lower bound. 

	To show that this bound is attained, let $P^*$ be a distribution that has the same marginal $X_i$ distribution as $P_0$, the same conditional distribution of $Y_i\mid X_i=x$ for $x$ such that $\pi_b(x)>0$, and has 
	\begin{equation}
		\mu_{P^*}(x)=\ell\vee \sup_{x':\pi_b(x')>0} \mu_{P_0}(x')-Ld(x, x')
	\label{eq:p_star}
	\end{equation}
	for $x$ such that $\pi_b(x)=0$. Furthermore, we assume that, for $x$ such that $\pi_b(x)=0$ the distribution $Y_i\mid X_i=x$ is a point mass at $\mu_{P^*}(x)$. This ensures that $P^*\in \mathcal{P}^M_{\bdd}$. 
	This $P^*$ clearly attains the bound of $\psi_2^{-}$ so it is sufficient to check that it is consistent with $P_0$ and $L$-Lipschitz to show that $P^*\in \mathcal{P}_L^{\Lip}$. 

	The distribution $P^*$ is consistent with $P_0$ by construction, but we must still check that $\mu_{P^*}$ defined in \eqref{eq:p_star} agrees with $\mu_{P_0}$ for $x$ in the overlap region. To verify this, note that for any $x$ such that $\pi_b(x)>0$, the supremum in \eqref{eq:p_star} is attained at $x'=x$, since $\mu_{P_0}(x)-Ld(x,x)=\mu_{P_0}(x)$ and $\mu_{P_0}(x)\geq \mu_{P_0}(x')-Ld(x,x')$ for any other $x'$ since $\mu_{P_0}$ is $L$-Lipschitz. So the supremum is in fact equal to $\mu_{P_0}(x)$, and $\ell\vee \mu_{P_0}(x)=\mu_{P_0}(x)$, which is consistent with $P_0$. 

	To check that $\mu_{P^*}$ is $L$-Lipschitz, let $f(x)$ be the value attaining the supremum in \eqref{eq:p_star}, so that $\mu_{P^*}(x)=\ell\vee (\mu_{P_0}(f(x))-Ld(x, f(x)))$. Then, consider any pair of points $x_1, x_2$. We distinguish three cases.

	First, if $\mu_{P_0}(f(x_1))-Ld(x_1, f(x_1))<\ell$ and $\mu_{P_0}(f(x_2))-Ld(x_2, f(x_2))<\ell$, we have $\mu_{P^*}(x_1)-\mu_{P^*}(x_2)=0$ and the Lipschitz condition is satisfied.

Second, if $\mu_{P_0}(f(x_1))-Ld(x_1, f(x_1))>\ell$ and $\mu_{P_0}(f(x_2))-Ld(x_2, f(x_2))>\ell$, assume without the loss of generality that $\mu_{P^*}(x_1)>\mu_{P^*}(x_2)$. Then we have
	\begin{align}
		|\mu_{P^*}(x_1)-\mu_{P^*}(x_2)|&= \mu_{P^*}(x_1)-\mu_{P^*}(x_2),\\
		&= \left(\mu_{P_0}(f(x_1))-Ld(x_1, f(x_1))\right)-\left( \mu_{P_0}(f(x_2)-Ld(x_2, f(x_2)) \right),\\
		&\leq \left(\mu_{P_0}(f(x_1))-Ld(x_1, f(x_1))\right)-\left( \mu_{P_0}(f(x_1)-Ld(x_2, f(x_1)) \right),\\
			&= L\left( d(x_1, f(x_1))-d(x_2, f(x_1)) \right),\\
			&\leq Ld(x_1, x_2),
	\end{align}
	where the first inequality uses the fact that $f(x_2)$ attains the supremum and the second uses the triangle inequality.

	Finally, if $\mu_{P_0}(f(x_1))-Ld(x_1, f(x_1))>\ell$ and $\mu_{P_0}(f(x_2))-Ld(x_2, f(x_2))<\ell$, we have 
\begin{align}
		|\mu_{P^*}(x_1)-\mu_{P^*}(x_2)|&= \left(\mu_{P_0}(f(x_1))-Ld(x_1, f(x_1))\right)-\left( \ell \right),\\
		&\leq \left(\mu_{P_0}(f(x_1))-Ld(x_1, f(x_1))\right)-\left( \mu_{P_0}(f(x_2)-Ld(x_2, f(x_2)) \right),\\
			&\leq Ld(x_1, x_2),
	\end{align}
	by the same arguments as above.

	This shows that $\mu_{P^*}$ is $L$-Lipschitz, and so $P^*\in \mathcal{P}_{L}^{\Lip}$. Thus the bound of $\psi_2^{-}$ is attained, completing the proof.
\end{proof}

\subsubsection{Proof of Theorem~\ref{thm:mu_hat_lip_gen}(c)}

Before proceeding to the proof of Theorem~\ref{thm:mu_hat_lip_gen}(c), we present two helpful lemmas. The first of these controls the difference between the suprema of interest in terms of the difference in conditional mean functions, and will essentially be used to replace the estimated $\hat{\mu}$ in $\hat{\psi}_2^-$ with the true $\mu$.

\begin{lemma}
	\label{lem:sup_diff}
	For any conditional mean functions $\mu_1, \mu_2$, we have 
	\begin{equation}
		\left| \max_{j: \pi_b(X_j)>0} \mu_1(X_j)-Ld(X_i, X_j)-\max_{j:\pi_b(X_j)>0} \mu_2(X_j)-Ld(X_i,x)\right|\leq \sup_{x:\pi_b(x)>0} |\mu_1(x)-\mu_2(x)|
	\end{equation}
\end{lemma}

\begin{proof}[Proof of Lemma~\ref{lem:sup_diff}]
	We have 
\begin{align}
		\max_{j:\pi_b(X_j)>0} \mu_1(X_j)-Ld(X_i, X_j)&\leq \max_{j:\pi_b(X_j)>0} \mu_1(X_j)-\mu_2(X_j)+\max_{j:\pi_b(X_j)>0} \mu_2(X_j)-Ld(X_i, X_j),\\
		&\leq \sup_{x: \pi_b(x)>0} |\mu_1(x)-\mu_2(x)|+\max_{j:\pi_b(X_j)>0} \mu_2(X_j)-Ld(X_i, X_j),\\
		\label{eq:j_star_vs_sup}
	\end{align}
	The same bounds hold with the roles of $\mu_1$ and $\mu_2$ reversed, completing the proof.
\end{proof}

The next two lemmas will allow us to replace the maximum in $\hat{\psi}_2^-$ by a supremum by controlling the difference between the maximum and supremum.

\begin{lemma}
	\label{lem:max_sup_bound}
	Let $x^*$ be the point at which $\sup_{x: \pi_b(x)>0} \mu_{P_0}(x)-Ld(X_i, x)$ is attained and let 
\begin{equation}
		{j^*}=\argmin_{j: \pi_b(X_j)>0} d(X_j, x^*)
\end{equation}
be the observed point that is closest to $x^*$. Then 
\begin{equation}
	\left|\max_{j:\pi_b(X_j)>0} \mu_{P_0}(X_j)-Ld(X_i, X_j)-\sup_{x: \pi_b(x)>0} \mu_{P_0}(x)-Ld(X_i, x)\right|\leq 2Ld(x^*, X_{j^*})
\end{equation}
\end{lemma}

\begin{proof}
We have 
\begin{align}
\mu_{P_0}(X_{j^*})-Ld(X_i, X_{j^*}) &\geq \mu_{P_0}(x^*)-Ld(X_{j^*}, x^*)-Ld(X_i, X_{j^*}),\\
	&\geq\mu_{P_0}(x^*)-Ld(X_i, x^*)-2Ld(x^*, X_{j^*}),\\
	&= \sup_{x: \pi_b(x)>0} \mu_{P_0}(x)-Ld(X_i, x)-2Ld(x^*, X_{j^*}),
\end{align}
where the first inequality uses the fact that $\mu_{P_0}$ is $L$-Lipschitz and the second is the triangle inequality. Since the supremum is greater than the maximium, the result follows.
\end{proof}

\begin{lemma}
	\label{lem:nn_convergence}
	Fix a point $x$ such that $\pi_b(x)>0$ and let 
\begin{equation}
	{j^*}=\argmin_{j: \pi_b(X_j)>0} d(X_j, x)
\end{equation}
be the index of the closest observation to $x$ in the overlap region. Then, under either of Assumption~\ref{ass:density}(a) or Assumption~\ref{ass:density}(b), we have $\P(d(x, X_{j^*})>\epsilon)\xrightarrow{n\to \infty} 0$ for any $\epsilon>0$.
\end{lemma}

\begin{proof}
We analyze $\P(d(x, X_{j^*})>\epsilon)$ under each of the two possibilities in Assumption~\ref{ass:density}. If Assumption~\ref{ass:density}(a) holds and there is an atom at $x$ so that $\P(X_i=x)=p$ for some $p>0$, we have 
\begin{align}
	\P(d(x, X_{j^*})>\epsilon)&\leq \P(X_i\neq x \text{ for }1\leq i\leq n),\\
	&=(1-p)^n,
\end{align}
which goes to 0 as $n\to \infty$. Next, if Assumption~\ref{ass:density}(b) holds, let $\delta$ be such that $\P(d(x, X_{j^*})\leq \epsilon)>\delta$. Then 
\begin{align}
	\P(d(x, X_{j^*})>\epsilon)&=\P(d(x, X_i)>\epsilon \text{ for }1\leq i\leq n),\\
	&\leq (1-\delta)^n
\end{align}
which again goes to 0 as $n\to \infty$. 
\end{proof}

With these lemmas in hand, we begin the main proof.

\begin{proof}[Proof of Theorem~\ref{thm:mu_hat_lip_gen}(b)]
	We begin by defining 
	\begin{equation}
		\hat{\psi}_2^{-,\text{oracle}}=\frac{1}{n}\sum_{i=1}^n \pi_e(X_i)\left( \ell\vee \sup_{x: \pi_b(x)>0} {\mu}_{P_0}(x)-Ld(X_i, x) \right)\mathbf{1}\{\pi_b(X_i)=0\}.
	\end{equation}
	The main idea is to show that 
	\begin{equation}
		|\hat{\psi}_2^- - \hat{\psi}_2^{-,\text{oracle}}|=o_P(1),
		\label{eq:oracle_equivalence}
	\end{equation}
	and then since $\hat{\psi}_2^{-,\text{oracle}}$ is a sum of i.i.d.~terms, the desired result follows by the law of large numbers. 

	To show \eqref{eq:oracle_equivalence}, we show that for any fixed $i$,
	\begin{equation}
		\left|\ell\vee \max_{j:\pi_b(X_j)>0}\hat{\mu}(X_j)-Ld(X_i, X_j)-\left(\ell\vee \sup_{x: \pi_b(x)>0} \mu_{P_0}(x)-Ld(X_i, x)\right)\right|=o_P(1).
		\label{eq:max_hat_vs_sup_mu}
	\end{equation}
	Indeed, we have by the triangle inequality that
	\begin{align}
		&\left|\ell\vee \max_{j:\pi_b(X_j)>0} \hat{\mu}(X_j)-Ld(X_i, X_j)-\left(\ell\vee \sup_{x: \pi_b(x)>0} \mu_{P_0}(x)-Ld(X_i, x)\right)\right|\\
		\leq &\left|\ell\vee \max_{j:\pi_b(X_j)>0} \hat{\mu}(X_j)-Ld(X_i, X_j)- \left(\ell\vee \max_{j:\pi_b(X_j)>0} \mu_{P_0}(X_j)-Ld(X_i, X_j)\right)\right|\label{eq:max_hat_vs_max_mu}\\
		&+\left|\ell\vee \max_{j:\pi_b(X_j)>0} {\mu}_{P_0}(X_j)-Ld(X_i, X_j)-\left(\ell \vee \sup_{x: \pi_b(x)>0} \mu_{P_0}(x)-Ld(X_i, x)\right)\right|, \label{eq:max_mu_vs_sup_mu}
	\end{align}
	and now we analyze \eqref{eq:max_hat_vs_max_mu} and \eqref{eq:max_mu_vs_sup_mu} separately. In both cases, we will show that if one of the terms is smaller than $\ell$, the other must be as well with high probability, and then work on the event that they are both smaller than $\ell$. 

	For \eqref{eq:max_hat_vs_max_mu}, we first consider the event
	\begin{equation}
		E=\left\{ \max_j \hat{\mu}(X_j)-Ld(X_i,X_j)>\ell, \max_j \mu_{P_0}(X_j)-Ld(X_i, X_j)<\ell \right\}.
	\end{equation}
	On the event $E$, there must exist some $j$ and $\epsilon>0$ such that $\hat{\mu}(X_j)-Ld(X_i,X_j)>\ell+\epsilon$. Then,
	\begin{align}
		\P(E)&\leq \P(\hat{\mu}(X_j)-Ld(X_i, X_j)>\ell+\epsilon, \mu_{P_0}(X_j)-Ld(X_i, X_j)<\ell),\\
		&\leq \P\left(\sup_{x:\pi_b(x)>0} | \hat{\mu}(x)-\mu_{P_0}(x)|>\epsilon\right),\\
		&= o(1),
	\end{align}
	by Assumption~\ref{ass:consistency}. 

	Since $\P(E)=o(1)$, it is sufficient to work on the event $E^c$. On this event, there are two possibilities. Either
	\begin{equation}
		\mu_{P_0}(X_j)-Ld(X_i, X_j)<\ell\quad\text{ and }\quad \sup_{x: \pi_b(x)>0} \mu_{P_0}(x)-Ld(X_i, x)<\ell,
		\label{eq:both_small}
	\end{equation}
	or
\begin{equation}
	\mu_{P_0}(X_j)-Ld(X_i, X_j)>\ell\quad\text{ and }\quad \sup_{x: \pi_b(x)>0} \mu_{P_0}(x)-Ld(X_i, x)>\ell.
		\label{eq:both_big}
	\end{equation}

	If \eqref{eq:both_small} holds, then \eqref{eq:max_hat_vs_max_mu} is 0. On the other hand, if \eqref{eq:both_big} holds, then \eqref{eq:max_hat_vs_max_mu} is $o_P(1)$ by Lemma~\ref{lem:sup_diff} and Assumption~\ref{ass:consistency}. Thus we conclude that \eqref{eq:max_hat_vs_max_mu} is $o_P(1)$ on the event $E^c$.

	For \eqref{eq:max_mu_vs_sup_mu}, we distinguish cases. For the first case, if 
	\begin{equation}
		\sup_{x:\pi_b(x)>0} \mu_{P_0}(x)-Ld(X_i, x)\leq \ell,
	\end{equation}
	then we must have
	\begin{equation}
		\max_{j:\pi_b(X_j)>0} \mu_{P_0}(X_j)-Ld(X_i, X_j)\leq \ell,
	\end{equation}
	as well, and so \eqref{eq:max_mu_vs_sup_mu} is 0. 

	Thus it suffices to consider the case where
	\begin{equation}
		\sup_{x:\pi_b(x)>0} \mu_{P_0}(x)-Ld(X_i, x)>\ell.
		\label{eq:sup_greater_ell}
	\end{equation}
	In this case, suppose that the supremum is attained at $x^*$, and that $\mu_{P_0}(x^*)-Ld(X_i,x^*)-\ell=\epsilon$ for some $\epsilon>0$, and let 
	\begin{equation}
		{j^*}=\argmin_{j: \pi_b(X_j)>0} d(X_j, x^*)
	\end{equation}
	be the index of the observed data point that is closest to $x^*$. 
	
	Then, consider the event 
	\begin{equation}
		E=\left\{ \max_{j:\pi_b(X_j)>0} \mu_{P_0}(X_j)-Ld(X_i,X_j)<\ell  \right\}.
	\end{equation}
Since we are working in the case where \eqref{eq:sup_greater_ell} holds, by Lemma~\ref{lem:max_sup_bound}, if the event $E$ occurs as well, we must have $2L d(x^*, X_{j^*})>\epsilon$. Thus $\P(E)\leq \P(2Ld(x^*, X_{j^*})>\epsilon)$, and $\P(2Ld(x^*, X_{j^*})>\epsilon)=o(1)$ by Lemma~\ref{lem:nn_convergence}, so we see that $\P(E)=o(1)$. 

Thus it suffices to work on the event $E^c$. On $E^c$, we have
\begin{equation}
	\max_{j:\pi_b(X_j)>0} \mu_{P_0}(X_j)-Ld(X_i, X_j) \geq \sup_{x: \pi_b(x)>0} \mu_{P_0}(x)-Ld(X_i, x)-2Ld(x^*, X_{j^*}),
\end{equation}
by Lemma~\ref{lem:max_sup_bound}. Applying Lemma~\ref{lem:nn_convergence} again, we see that \eqref{eq:max_mu_vs_sup_mu} is $o_P(1)$.

Since \eqref{eq:max_hat_vs_max_mu} and \eqref{eq:max_mu_vs_sup_mu} are both $o_P(1)$, we conclude that \eqref{eq:max_hat_vs_sup_mu} holds, implying \eqref{eq:oracle_equivalence} and finishing the proof.
\end{proof}

\subsection{Proof of Theorem 3}

\begin{proof}
	We follow the approach of the proof of Theorem~\ref{thm:mu_hat_lip_gen}(b), and begin by defining 
\begin{equation}
	\hat{\psi}_2^{-,\text{oracle}}=\frac{1}{n}\sum_{i=1}^n \pi_e(X_i)\left(\sup_{x: \pi_b(x)>0} {\mu}_{P_0}(x)-Ld(X_i, x) \right)\mathbf{1}\{\pi_b(X_i)=0\}.
	\end{equation}
We then decompose 
\begin{align}
	\E[(\hat{\psi}_2^- - \psi_2^{-})^2]&= \E[(\hat{\psi}_2^--\hat{\psi}_2^{-,\text{oracle}}+\hat{\psi}_2^{-,\text{oracle}}-\psi_2^{-})^2],\\
	&\leq 2\E[(\hat{\psi}_2^--\hat{\psi}_2^{-,\text{oracle}})^2]+2\E[(\hat{\psi}_2^{-,\text{oracle}}-\hat{\psi}_2^{-})^2].
	\label{eq:oracle_split}
\end{align}
The second term of \eqref{eq:oracle_split} is the variance of an i.i.d.~sum, and is thus equal to $C/n$ for some constant $C$. The remainder of the proof focuses on the first term of \eqref{eq:oracle_split}.

The first term of \eqref{eq:oracle_split} is
\begin{align}
	&\E\left[ \left( \frac{1}{n}\sum_{i=1}^n \pi_e(X_i)\left( \max_{j>\pi_b(X_j)>0} \hat{\mu}(X_j)-Ld(X_i, X_j)-\sup_{x: \pi_b(x)>0} \mu_{P_0}(x)-Ld(X_i, x)\right)\mathbf{1}\{\pi_b(X_i)=0\}  \right)^2 \right],\\
	&\leq \frac{1}{n}\sum_{i=1}^n \E\left[ \left( \max_{j>\pi_b(X_j)>0} \hat{\mu}(X_j)-Ld(X_i, X_j)-\sup_{x: \pi_b(x)>0} \mu_{P_0}(x)-Ld(X_i, x)\ \right)^2 \right],\label{eq:cs_sum}
\end{align}
by the Cauchy-Schwarz inequality and the fact that $0\leq \pi_e, \mathbf{1}\{\pi_b(X_i)=0\}\leq 1$. A single term of \eqref{eq:cs_sum} is
\begin{align}
	\leq &2\E\left[ \left( \max_{j: \pi_b(X_j)>0} \hat{\mu}(X_j)-Ld(X_i, X_j)  - \max_{j: \pi_b(X_j)>0} \mu_{P_0}(X_j)-Ld(X_i, X_j) \right)^2 \right]\\ &+2\E\left[ \left( \max_{j:\pi_b(X_j)>0} \mu_{P_0}(X_j)-Ld(X_i, X_j)- \sup_{x: \pi_b(x)>0} \mu_{P_0}(x)-Ld(X_i, x)  \right)^2 \right],\\
	\leq&2\E[\|(\hat{\mu}(x)-\mu(x))\mathbf{1}\{\pi_b(x)>0\}\|_{\infty}^2]+2\E[L^2d(x^*, X_{j^*})^2],\\
	\leq&2\E[\|(\hat{\mu}(x)-\mu(x))\mathbf{1}\{\pi_b(x)>0\}\|_{\infty}^2]+4L^2(c_dbn)^{-2/d}.
\end{align}
where the second inequality uses Lemmas~\ref{lem:sup_diff} and \ref{lem:max_sup_bound} (here $x^*$ and $X_{j^*}$ are as defined in Lemma~\ref{lem:max_sup_bound}), and the third inequality uses Lemma~\ref{lem:expected_dist} (proven in Section~\ref{subsec:lemmas}).

\end{proof}

\subsection{Proof of Theorem 4}

\begin{proof}
	Our proof relies on LeCam's two point method \citep{lecam1973, wainwright2019}, which states that if we can find distributions $P_1$ and $P_2$ in $\mathcal{P}_L^{\Lip}$ with $|\psi_2^{-}(P_1)-\psi_2^{-}(P_2)|\geq 2\delta$, then 
\begin{equation}
	\inf_{\hat{\psi}_2^-} \sup_{P\in \mathcal{P}_L^{\Lip}} \E_P\left[ (\hat{\psi}-\psi_2^{-}(P))^2 \right]\geq \frac{\delta^2}{2}\left( 1-\|P_1^n-P_2^n\|_{\text{TV}} \right),
\end{equation}
where by a slight abuse of notation we write $P_1^n$ for the joint distribution of $(X_1, A_1, A_1Y_1),\cdots, (X_n, A_n, A_nY_n)$ when $(X_i, Y_i)$ are drawn from $P_1$, and similarly for $P_2$. We construct $P_1$ and $P_2$ as follows: 
\begin{enumerate}[(i)]
	\item under $P_1$, the marginal distribution of $X_i$ is uniform on $[-1,1]^d$, and $Y_i\mid X_i=x$ is $N(\mu_1(x),1)$ where $\mu_1(x)=0$ identically 
	\item under $P_2$, the marginal distribution of $X_i$ is uniform on $[-1,1]^d$, and the distribution of $Y_i\mid X_i=x$ is $N(\mu_2(x), 1)$ where 
		\begin{equation}
			\mu_2(x)=\left\{\begin{aligned} &0&&\text{ if }&&&x\in [-1/2+L\epsilon, 1/2-L\epsilon]^d\\ &\tilde{\mu}_2(x) &&\text{ if }&&&x\in [-1/2-L\epsilon, 1/2+L\epsilon]^d\setminus [-1/2+L\epsilon, 1/2-L\epsilon]^d,\\ &0&&\text{ if }&&& [-1,1]^d\setminus [-1/2-L\epsilon, 1/2+L\epsilon]^d\end{aligned}\right.,
		\end{equation}
		where $\epsilon<1/(2L)$ is arbitrary and $\tilde{\mu}_2(x)$ is a function that is equal to $\epsilon$ for $x$ on the boundary of the set $[-1/2, 1/2]^d$, and is equal to 0 on the boundaries of $[-1/2-\epsilon, 1/2+\epsilon]^d$ and $[-1/2+\epsilon, 1/2-\epsilon]^d$, and linearly interpolates between these values. 
\end{enumerate}
Essentially, the distribution of $Y_i\mid X_i$ under $P_2$ has a bump of size $\epsilon$ at the boundary of the support of $\pi_b$ that decays to zero as fast as the Lipschitz assumption will allow. 

We now proceed in two steps: first, we verify the separation condition of LeCam's lemma; second, we upper bound the distance between $P_1^n$ and $P_2^n$. 

\paragraph{Checking the separation condition.} For the separation condition, we begin by noting that 

\begin{align}
	\psi_2^{-}(P_1)&= \E\left[ \pi_e(X_i)\left( \sup_{x: \pi_b(x)>0} \mu_1(x)-Ld(X_i, x) \right)\mathbf{1}\{\pi_b(X_i)=0\} \right],\\
	&= \E\left[ \pi_e(X_i)\left( \sup_{x: \pi_b(x)>0} -Ld(X_i, x) \right)\mathbf{1}\{\pi_b(X_i)=0\} \right]\label{eq:first_point_parameter}
\end{align}
Suppose that for each $X_i$, $\sup_{x:\pi_b(x)>0} -Ld(X_i, x)$ is attained at a point $f(X_i)$ and note that $f(X_i)$ must lie on the boundary of the cube $[-1/2, 1/2]^d$ (since it is the projection of $X_i$ onto the support of $\pi_b$). Then,
\begin{align}
	\psi_2^{-}(P_2)&= \E\left[ \pi_e(X_i)\left( \sup_{x: \pi_b(x)>0} \mu_2(x)-Ld(X_i, x) \right)\mathbf{1}\{\pi_b(X_i)=0\} \right],\\
	&\geq \E\left[ \pi_e(X_i)\left( \mu_2(f(X_i))-Ld(X_i, f(X_i)) \right)\mathbf{1}\{\pi_b(X_i)=0\} \right],\\
	&= \epsilon \E[\pi_e(X_i)\mathbf{1}\{\pi_b(X_i)=0\}]+\E\left[ \pi_e(x)(-Ld(X_i, f(X_i)))\mathbf{1}\{\pi_b(X_i)=0\} \right],\\
	&= \epsilon\E[\pi_e(X_i)\mathbf{1}\{\pi_b(X_i)=0\}]+\psi_2^{-}(P_1)\label{eq:separation},
\end{align}
where the first inequality bounds the supremum by a particular point, the second equality uses the fact that $\mu_2$ is equal to $\epsilon$ on the boundary of $[-1/2, 1/2]^d$, and the third uses the fact that $f$ attains the supremum in \eqref{eq:first_point_parameter}. Then, it follows from \eqref{eq:separation} that the separation condition is satisfied with $\delta=\epsilon\E[\pi_e(X_i)\mathbf{1}\{\pi_b(X_i)=0\}]/2$.

\paragraph{Bounding the total variation distance.} It remains to bound $\|P_1^n-P_2^n\|_{\text{TV}}$. By Pinsker's inequality, we have 
\begin{align}
	\|P_1^n-P_2^n\|_{\text{TV}}&\leq \sqrt{\frac{1}{2}D_{\text{KL}}(P_1^n\|P_2^n)},\\
	&\leq \sqrt{\frac{n}{2} D_{\text{KL}}(P_1\|P_2)},\\
\end{align}
Now, letting $p_1$ and $p_2$ be the densities of $P_1$ and $P_2$ respectively, we have
\begin{equation}
	D_{\text{KL}}(P_1\|P_2)= \int_{[-1,1]^d}\int_{-\infty}^{\infty} \sum_{a\in \{0,1\}}p_1(x,a,y)\log \frac{p_1(x,a,y)}{p_2(x,a,y)}\, dy\,dx.
	\label{eq:kl}
\end{equation}
Note that 
\begin{align}
	\sum_{a\in \{0,1\}} p_1(x,a,y)\log \frac{p_1(x,a,y)}{p_2(x,a,y)}=& p_1(x)(1-\pi_b(x))p_1(y\mid x)\log\frac{p_1(x)(1-\pi_b(x))p_1(y\mid x)}{p_2(x)(1-\pi_b(x))p_2(y\mid x)}\\
	&+p_1(x)\pi_b(x)p_1(y\mid x)\log \frac{p_1(x)\pi_b(x)p_1(y\mid x)}{p_2(x)\pi_b(x)p_2(y\mid x)},\\
	=&\frac{1}{2^d}p_1(y\mid x)\log \frac{p_1(y\mid x)}{p_2(y\mid x)},
\end{align}
since $p_1(x)=p_2(x)=1/2^d$ for all $x$. Thus, \eqref{eq:kl} is 
\begin{align}
	&=\int_{[-1,1]^d}\int_{-\infty}^{\infty} 2^{-d}p_1(y\mid x)\log \frac{p_1(y\mid x)}{p_2(y\mid x)}\, dy\,dx\\
	&= \int_{[-1,1]^d} D_{\text{KL}}(N(0, 1)\| N(\mu_2(x), 1))\, dx,\\
	&= \int_{[-1,1]^d} \frac{1}{2}\mu_2(x)^2\, dx,\\
	&= \int_{[-1,1]^d} \tilde{\mu}_2(x)^2\mathbf{1}\{x\in [-1/2-L\epsilon, 1/2+L\epsilon]^d\setminus [-1/2+L\epsilon, 1/2-L\epsilon]^d \}\, dx,\\
	&\leq \epsilon^2 ((1+2L\epsilon)^d-(1-2L\epsilon)^d)\label{eq:kl_bound},\\
	&\leq \epsilon^2(4L\epsilon)^d,
\end{align}
where the last line uses the fact that $(1+x)^d-(1-x)^d\leq (1+x)^d\leq (2x)^d$ for $0<x<1$.

The calculations above show that $\|P_1^n-P_2^n\|_{\text{TV}}^2\leq \frac{n}{2}(4L)^d\epsilon^{d+2}$. If we set $\epsilon=(2n)^{-1/(d+2)}(4L)^{-d/(d+2)}$, this gives the bound $\|P_1^n-P_2^n\|\leq 1/2$. This choice of $\epsilon$ then gives the lower bound 
\begin{align}
	\inf_{\hat{\psi}_2^-} \sup_{P\in \mathcal{P}_L^{\Lip}} \E_P\left[(\hat{\psi}_2^- - \psi_2^{-}(P))^2  \right]&\geq \frac{\delta^2}{4},\\
	&=\frac{1}{16}\E[\pi_e(X_i)\mathbf{1}\{\pi_b(X_i)=0\}]^2\epsilon^2,\\
	&= \frac{1}{16}\E[\pi_e(X_i)\mathbf{1}\{\pi_b(X_i)=0\}]^2(2n)^{-2/(d+2)}(4L)^{-2/(d+2)},\\
\end{align}
as long as the condition that $\epsilon < 1/(2L)$ is satisfied. We can check that this condition holds whenever $n\geq 2^{-d+1}L^2$, completing the proof. 
\end{proof}

\subsection{Proof of Lemma~\ref{lem:expected_dist}}
\label{subsec:lemmas}

In this section, we prove a technical lemma used in the main proofs. 

\begin{lemma}
	\label{lem:expected_dist}
	Let $x^*$ and $X_{j^*}$ be as in Lemma~\ref{lem:max_sup_bound} and let $b$ be a lower bound on the density of $X_i$. Then $\E[d(x^*, X_{j^*})^2]\leq 2(c_dbn)^{-2/d}$, where $n$ is the sample size, $d$ is the dimension of the covariate space, and $c_d$ is the volume of the unit ball in $d$-dimensions.
\end{lemma}

\begin{proof}
	Let $c_d$ be the volume of the unit ball in $d$ dimensions. We compute 
	\begin{align}
		\E[d(x^*, X_{j^*})^2]&= \int_0^{\infty}\P(d(x^*, X_{j^*})\geq \sqrt{t})\, dt,\\
		&= \int_{0}^{\infty} \P(d(x^*, X_{j})\geq \sqrt{t})^n\, dt,\\
		&\leq \int_0^{\infty} (1-c_dt^{d/2}b)_+^n,\\
		&= (c_db)^{-2/d}\int_0^{\infty} (1-u^{d/2})_+^n\, du \tag{$u = t(c_db)^{2/d}$},\\
		&= (c_db)^{-2/d}\int_0^{1} (1-u^{d/2})^n\, du.\label{eq:integral}
	\end{align}
	If $d=1$, we can manually compute that \eqref{eq:integral} is 
	\begin{equation}
		\frac{2(c_db)^{-2/d}}{n^2+3n+2}\leq \frac{2(c_db)^{-2/d}}{n^2}=2(c_dbn)^{-2/d},
	\end{equation}
	and the result of the lemma is satsified. We now assume that $d>1$, and continue bounding by 
	\begin{align}
		&\leq (c_db)^{-2/d}\int_0^1 \exp(-u^{d/2}n)\, du,\\
		&\leq (c_dbn)^{-2/d}\int_0^{n^{2/d}} \exp(-v^{d/2})\, dv, \tag{$v=un^{2/d}$},\\
		&= (c_dbn)^{-2/d}\int_0^{n^{2/d}} \exp(-v^{d/2}),\\
		&=\ (c_dbn)^{-2/d}\left(\int_0^1 \exp(-v^{d/2})\, dv+\int_1^{n^{2/d}}\exp(-v^{d/2})\, dv\right),\\
		&\leq  (c_dbn)^{-2/d}\left(1+\int_1^{n^{2/d}}\exp(-v)\, dv\right),\\
		&\leq  (c_dbn)^{-2/d}(1+e^{-1}),\\
		&\leq 2(c_dbn)^{-2/d}.
	\end{align}
	So, in all cases, $\E[d(x^*, X_{j^*})^2]\leq 2(c_dbn)^{-d/2}$. 
\end{proof}

\section{Computational improvements on $\hat{\psi}_2^-$} 
\label{sec:nn}

In this section, we discuss approximations of
\begin{equation}
	\hat{\psi}_2^-=\frac{1}{n}\sum_{i=1}^n \pi_e(X_i)\left( \max_{j: \pi_b(X_j)>0} \hat{\mu}(X_j)-Ld(X_i, X_j) \right)\mathbf{1}\{\pi_b(X_i)=0\}
	\label{eq:app_psi2_hat}
\end{equation}
with favorable computational properties. The idea is to bound the maximum in \eqref{eq:app_psi2_hat} by the value of $\hat{\mu}(X_j)-Ld(X_i,X_j)$ for some particular $j$. A natural choice of $j$ is the index of the nearest neighbor of $X_i$ in the overlap region,
\begin{equation}
	\nn(i)= \argmin_{j: \pi_b(X_j)>0} d(X_i, X_j).
\end{equation}
Then, we define
\begin{equation}
	\hat{\psi}_2^{-,\text{cons}} = \frac{1}{n}\sum_{i=1}^n \pi_e(X_i)\left( \hat{\mu}(X_{\nn(i)})-d(X_i, X_{\nn(i)}) \right)\mathbf{1}\{\pi_b(X_i)=0\},
\end{equation}
as a conservative approximation of $\hat{\psi}_2^-$. By conservative, we mean that $\hat{\psi}_2^-\geq \hat{\psi}_2^{-,\text{cons}}$, so that the bounds obtained by using $\hat{\psi}_2^{-, \text{cons}}$ are always wider than those obtained by using $\hat{\psi}_2^-$. This ensures that the interval $[\hat{\psi}^-, \hat{\psi}^+]$ constructed using $\hat{\psi}_2^{-, \text{cons}}$ will still be consistent for $\psi(P_0)$ in the sense of Lemma~\ref{thm:interval_consistency}.

However, computing $\hat{\psi}_2^{-,\text{cons}}$ will typically be much faster than computing $\hat{\psi}_2^-$. This is because computing $\hat{\psi}_2^{-,\text{cons}}$ only requires finding the nearest overlap neighbor of each point in the no-overlap region, and then evaluating $\hat{\mu}(X_{\nn(i)})-d(X_i, X_{\nn(i)})$ for that nearest neighbor, rather than evaluating over all points in the overlap region and taking the maximum. This means that, once we have found the nearest overlap neighbor of each point in the no-overlap region, we need only $O(n)$ further computations to compute $\hat{\psi}_2^{-,\text{cons}}$. 

We now consider the complexity of nearest-neighbor search problem. To find an exact nearest-neighbor for each point in the no-overlap region generally requires computing all of the pairwise distances $d(X_i, X_j)$, and thus will still take time $O(n^2)$. However, we find in practice that since only $O(n)$ further computation is required to compute $\hat{\psi}_2^{-,\text{cons}}$, this approach is still faster than exactly computing $\hat{\psi}_2^-$. Furthermore, there exist data structures for exact nearest-neighbor search which use various heuristics to improve performance, and these can be leveraged for further computational gains \citep{bentley1975, omohundro1989}.

For settings that require a method that is faster than $O(n^2)$, there are two possibilities depending on the dimensionality of the covariates. If the covariates $X_i$ are one-dimensional, methods based on Voronoi diagrams can be used to compute exact nearest neighbors for all points in the overlap region in $O(n)$ time \citep{har2011, preparata2012}. In higher dimensions, we can instead use approximate nearest-neighbor search algorithms. For example, if the covariates $X_i$ lie in $\R^p$, and the metric $d$ is Euclidean distance, there exist algorithms that return a point $j^*$ such that $d(X_i, X_{j^*})\leq c d(X_i, X_{\nn(i)})$ in time $O(n^{1/(2c^2-1-o(1)}+pn^{o(1)}))$ \citep{andoni2018}. For a moderate value of $c$, such as $c=2$, this gives a runtime of $O(n^{1/(7-o(1)}+pn^{o(1)})$ for each point in the no-overlap region. Since there are $O(n)$ points in the no-overlap region, this leads to a total runtime of $O(n^{8/7+o(1)}+pn^{1+o(1)})$, improving significantly on the $O(n^2)$ time required when using exact nearest neighbors. Appealingly, even if we use an approximate nearest neighbor rather than an exact one, we will still obtain a conservative estimate of $\hat{\psi}_2^-$, and thus retain statistical validity.

	Finally, another computational advantage of $\hat{\psi}_2^{-,\text{cons}}$ is that once we have identified the nearest-neighbor of each point in the overlap region, we can compute $\hat{\psi}_2^{-,\text{cons}}$ for any value of $L$ with only $O(n)$ operations. This makes search over a large range of values of $L$, as is done for sensitivity analyses like those shown in Section~\ref{sec:experiments}, quite efficient as well. 


\section{Additional experimental results}
\label{sec:app_sims}

\subsection{Additional results for yeast dataset}
In Figure~\ref{fig:coverage}, we plot the results from Table~\ref{tab:yeast} for $L=3, 4, 5, \infty$, since these were the values for which we determined (based on the feasibility of the optimization problem) that the smoothness assumption was plausible. We see that, for these values of $L$, the coverage as defined in Theorem~\ref{thm:interval_consistency} with $\epsilon=0.01$ approaches 100\%.

In Table~\ref{tab:yeast_app}, we show the results of Table~\ref{tab:yeast} when defining coverage as in Theorem~\ref{thm:interval_consistency} with $\epsilon=0.005$ rather than $\epsilon=0.01$, as in the main text. We see that coverage rates are lower, but still approach 1 as the sample size increases. Crucially, the coverage for large values of $L$ remains comparable to the coverage of the Manski intervals, suggesting that the Lipschitz smoothness assumptions for those values of $L$ are plausible. 

\begin{figure}[h]
	\centering
	\includegraphics[width=0.7\linewidth]{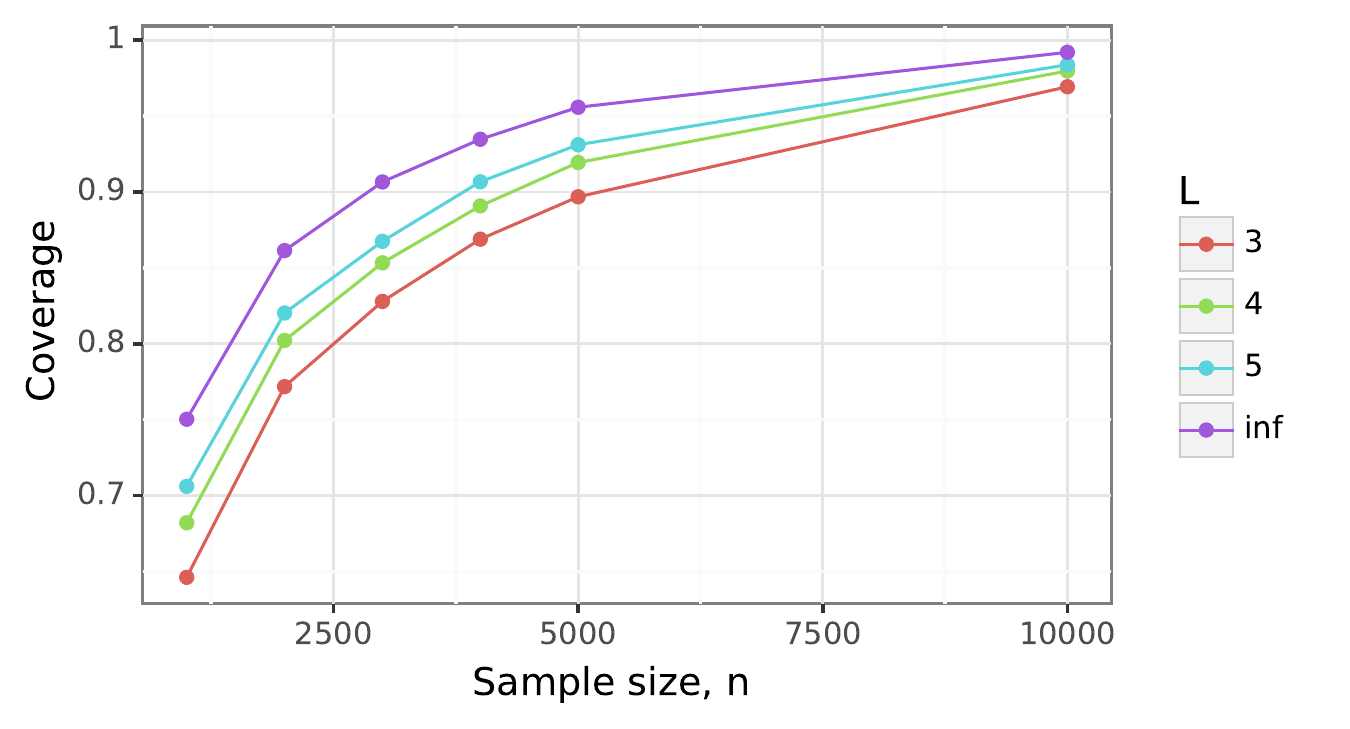}
	\caption{Visualization of results from Table~\ref{tab:yeast} for the values of $L$ for which the optimization problem is consistently feasible and thus the smoothness assumption is plausible. We see that as $n\to \infty$, the coverage (as defined in Theorem~\ref{thm:interval_consistency} with $\epsilon = 0.01$) approaches 100\%.}
	\label{fig:coverage}
\end{figure}

\begin{table}
	\centering
	\caption{The same experiment as in Table~\ref{tab:yeast}, but with coverage defined using $\epsilon=0.005$ in Theorem~\ref{thm:interval_consistency} rather than $\epsilon=0.01$. With the smaller value of $\epsilon$, coverage rates are lower, but still approach 1 as the sample size increases. Furthermore, the coverage of the Manski intervals is comparable to the coverage of the Lipschitz intervals for large values of $L$, indicating that those smoothness assumptions are plausible.\vspace*{0.1in}}
	\resizebox{\linewidth}{!}{	
	\begin{tabular}{lllllll}
		\toprule
		$L$&$n=1000$&$n=2000$&$n=3000$&$n=4000$&$n=5000$&$n=10000$\\
		\midrule 
		1&0.238 (0.785)&0.0002 (0.002)&0 (0.00)&0 (0.00)&0 (0.00)&0 (0.00)\\
		2&0.392 (1.00)&0.493 (1.00)&0.543 (1.00)&0.576 (1.00)&0.603 (0.995)&0.060 (0.106)\\
		3&0.450 (1.00)&0.561 (1.00)&0.613 (1.00)&0.643 (1.00)&0.676 (1.00)&0.766 (1.00)\\
		4&0.498 (1.00)&0.6057 (1.00)&0.657 (1.00)&0.685 (1.00)&0.723 (1.00)&0.810 (1.00)\\
		5&0.528 (1.00)&0.632 (1.00)&0.687 (1.00)&0.716 (1.00)&0.752 (1.00)&0.841 (1.00)\\
		$\infty$&0.591 (1.00)&0.706 (1.00)&0.757 (1.00)&0.789 (1.00)&0.819 (1.00)&0.895 (1.00)\\
		\bottomrule
	\end{tabular}
	}
	\label{tab:yeast_app}
\end{table}

\subsection{Additional results for Yahoo! Front Page Today dataset}

In Figure~\ref{fig:yahoo_ope_appendix}, we repeat the experiment on the Yahoo!\ Front Page Today dataset described in Section~\ref{sec:experiments} for a wider range of smoothness parameters $L$ and cutoffs $T$. With this range of values, we make two new observations: first, the convergence of our interval to the Manski interval in the upper endpoint was not apparent for small values of $L$, but is apparent for the values of $L$ considered here. Second, when $T=0.5$, there are no longer any overlap violations, and our partial identification intervals have length zero and recover the IPW estimator $\hat{\psi}_1$.

\begin{figure}[t]
	\centering
	\includegraphics[width=\textwidth]{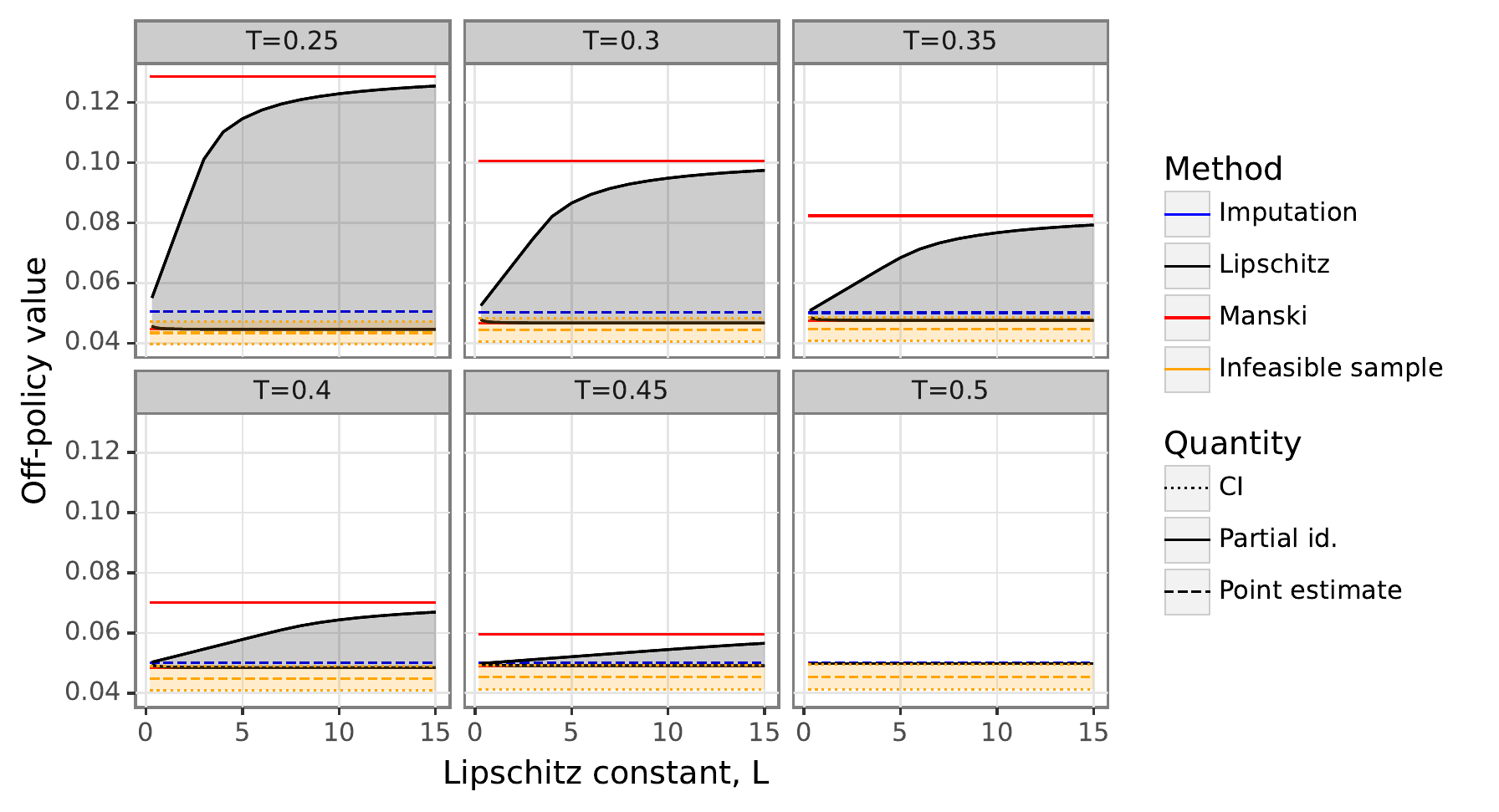}
	\caption{The same experiment as in Figure~\ref{fig:yahoo_ope_manski} over a wider range of values of $T$ and $L$. With this range of values, we see convergence of our bounds to the Manski bounds as $L$ grows large, and also that when $T=0.5$ and there are no longer any overlap violations, our intervals have width zero as expected.}
	\label{fig:yahoo_ope_appendix}
\end{figure}


\section{Counterexample for no-interaction under smoothness and monotonicity}
\label{sec:counter}

In this section, we present an example showing that the no-interaction property fails if we make both smoothness and monotonicity assumptions. 

\begin{figure}[h!]
	\centering
	\includegraphics{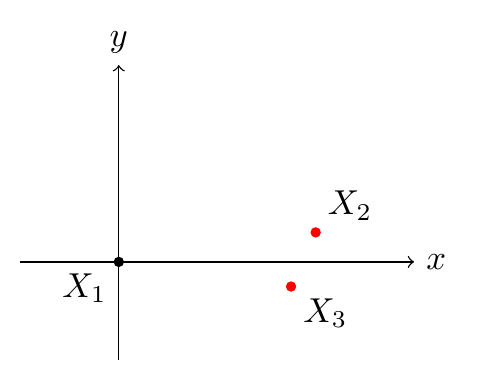}
	\caption{Counterexample showing that the no-interaction condition fails when making both a Lipschitz smoothness assumption and a monotonicity assumption. Here, $X_1$ lies in the overlap region, but $X_2$ and $X_3$ do not, and the only ordering between the points is that $X_1\prec X_2$. Thus, the tight bound that the $i=1$ unit implies on the $i=2$ unit is $t_2\geq \hat{\mu}(X_1)$, and propagating this to the $i=3$ unit gives that $t_3\geq \hat{\mu}(X_1)-Ld(X_2, X_3)$. This will be tighter than the bound directly coming from the $i=1$ unit, $t_3\geq \hat{\mu}(X_1)-Ld(X_1,X_3)$, whenever we have $d(X_2, X_3) < d(X_1, X_3)$, as in the figure, and so there is interaction between points in the no-overlap region.}
	\label{fig:counter}
\end{figure}

To construct our example, we consider a problem with $n=3$ points with two-dimensional covariates $X_i$ such that $X_1=(0,0)$ is the origin, $X_2$ lies in the first quadrant, and $X_3$ lies in the fourth quadrant. We assume that $X_1$ is in the overlap region, and that $X_2$ and $X_3$ are not. This configuration is illustrated Figure~\ref{fig:counter}.

We will assume both that the conditional mean $\mu$ is $L$-Lipschitz and that it is monotone with respect to the dictionary order $\prec$. That is, we have $(x_1, y_1)\prec (x_2, y_2)$ if $x_1\leq x_2$ and $y_1\leq y_2$. 

With this set-up, we have that $t_1=\hat{\mu}(X_1)$ since $X_1$ lies in the overlap region. The constraints induced on the $t_2$ point by $t_1$ are 
\begin{equation}
	t_2\geq \hat{\mu}(X_1)-Ld(X_1, X_2)\quad\text{and}\quad t_2\geq \hat{\mu}(X_1),
\end{equation}
since $X_1\prec X_2$. The latter of these is always tighter, and so to minimize the objective we set $t_2=\hat{\mu}(X_1)$. With this choice of $t_2$, the bounds induced on $t_3$ by $t_1$ and $t_2$ are
\begin{equation}
	t_3\geq \hat{\mu}(X_1)-Ld(X_1, X_3)\quad\text{and}\quad t_3\geq \hat{\mu}(X_1)-Ld(X_2, X_3),
\end{equation}
respectively. If we select $X_2$ and $X_3$ so that $d(X_2, X_3)< d(X_1, X_3)$, then the bound induced on $t_3$ by $t_2$ will be tighter than the bound induced on $t_3$ by $t_1$, and so we see that the constraints in the no-overlap region do interact with each other in this case. 

This example thus highlights that our results in Section~\ref{sec:estimator} are non-trivial and highlight special properties of the smoothness assumption that have not been previously observed; further characterization of what kinds of assumptions and combinations satisfy the no-interaction property is an interesting direction for future work.

\section{Connections to results of \citet{ben2021}}
\label{sec:ben_compare}

In this section, we describe how our results build on and extend the work of \citet{ben2021}. In our notation, the procedure of \citet{ben2021} for policy learning with no overlap is given by the max-min optimization problem
\begin{equation}
	\argmax_{\pi \in \Pi} \min_{P \in \mathcal{P}_L^{\Lip}} \frac{1}{n}\sum_{i=1}^n \mu_P(X_i)\pi(X_i)\mathbf{1}\{\pi_b(X_i)=0\},
\end{equation}
for some policy class $\Pi$. To solve this problem, \citet{ben2021} observe that 
\begin{equation}
	\min_{P\in \mathcal{P}_L^{\Lip}} \frac{1}{n}\sum_{i=1}^n \mu_P(X_i)\pi(X_i)\mathbf{1}\{\pi_b(X_i)=0\} \geq \frac{1}{n}\sum_{i=1}^n \pi(X_i) \left( \max_{j:\pi_b(X_j)>0} \tilde{\mu}(X_j)-Ld(X_i, X_j \right)\mathbf{1}\{\pi_b(X_i)=0\},
		\label{eq:bm_result}
\end{equation}
where $\tilde{\mu}$ is a simultaneous lower confidence bound on the conditional mean function $\mu_{P_0}$, and then solve the more conservative problem 
\begin{equation}
	\argmax_{\pi \in \Pi} \frac{1}{n}\sum_{i=1}^n \pi(X_i) \left( \max_{j:\pi_b(X_j)>0} \tilde{\mu}(X_j)-Ld(X_i, X_j \right)\mathbf{1}\{\pi_b(X_i)=0\}
		\label{eq:bm_problem}
\end{equation}
to learn a policy.

For the purposes of policy learning, such a conservative lower bound is sufficient\textemdash indeed, if we were to subtract a large constant, say 100, from the right-hand side of \eqref{eq:bm_result}, the solution to \eqref{eq:bm_problem} would remain unchanged. However, for policy evaluation, such conservative bounds are not acceptable: we would like to use the exact value of the left-hand side of \eqref{eq:bm_result} as a lower bound, and not incur needlessly wide bounds. 

The calculation of this exact value is the contribution of our Theorem~\ref{thm:mu_hat_lip}(a), which shows that 
\begin{equation}
	\min_{P\in \mathcal{P}_L^{\Lip}} \frac{1}{n}\sum_{i=1}^n \mu_P(X_i)\pi(X_i)\mathbf{1}\{\pi_b(X_i)=0\} = \frac{1}{n}\sum_{i=1}^n \pi(X_i) \left( \max_{j:\pi_b(X_j)>0} \hat{\mu}(X_j)-Ld(X_i, X_j \right)\mathbf{1}\{\pi_b(X_i)=0\},
\end{equation}
where $\hat{\mu}$ is an estimate of $\mu_{P_0}$. Thus, the right-hand side of the previous display, which is exactly our $\hat{\psi}_2^-$, is the correct estimator of a bound on the off-policy value. This result, along with its proof based on the no-interaction property, is novel to our work, and guides practitioners on dealing with overlap violations in a policy evaluation problem while ensuring tight bounds.

\end{document}